\documentclass[conference,10pt]{IEEEtran}
\usepackage[utf8]{inputenc} 
\usepackage[T1]{fontenc}    
\usepackage{hyperref}       
\usepackage{url}            
\usepackage{booktabs}       
\usepackage{amsfonts}       
\usepackage{nicefrac}      
\usepackage{microtype}      
\usepackage{epsfig}
\usepackage{times}
\usepackage{float}
\usepackage{afterpage}
\usepackage{amsmath}
\usepackage{amstext}
\usepackage{amssymb,bm}
\usepackage{latexsym}
\usepackage{color}
\usepackage{graphicx}
\usepackage{amsthm}
\usepackage{wrapfig}
\usepackage[center]{caption}
\usepackage{pstricks}
\usepackage{booktabs}
\usepackage{multicol}
\usepackage{lipsum}
\usepackage{dblfloatfix}
\usepackage{tikz}
\usepackage{pgfplots}
\usepackage{subcaption}
\usepackage{algorithm} 
\usepackage[noend]{algpseudocode} 
\usepackage{savesym}
\savesymbol{checkmark}
\usepackage{dingbat}
\usepackage{tikz}
\usepackage{pgfplots}
\algnewcommand\algorithmicforeach{\textbf{for each}}

\algdef{S}[FOR]{ForEach}[1]{\algorithmicforeach\ #1\ \algorithmicdo}
\allowdisplaybreaks

\pgfplotsset{compat=1.3}
\usepackage[font=small,labelfont=bf]{caption}

\DeclareMathOperator*{\argmax}{arg\,max}
\DeclareMathOperator*{\argmin}{arg\,min}

\newcommand{\snr}{\mathsf{snr}}

\newcommand{\E}{\mathbb{E}}

\newtheorem{thm}{Theorem}

\newtheorem{lem}[thm]{Lemma}
\newtheorem{prop}[thm]{Proposition}
\newtheorem{remark}{Remark}
\newtheorem{defin}{Definition}

\definecolor{darkgreen}{rgb}{0.55, 0.71, 0.13}

\long\def\symbolfootnote[#1]#2{\begingroup%
	\def\thefootnote{\fnsymbol{footnote}}\footnote[#1]{#2}\endgroup} 
\title{The Most Informative Order Statistic and its Application to Image Denoising
}%
\author{
\IEEEauthorblockN{Alex Dytso$^{\star}$,  Martina Cardone$^*$,  Cynthia Rush$^\dagger$}
$^{\star}$ New Jersey Institute of Technology, Newark, NJ 07102, USA Email: alex.dytso@njit.edu\\
$^*$ University of Minnesota, Minneapolis, MN 55404, USA, Email: mcardone@umn.edu\\
$^\dagger$ Columbia University, New York, NY 10025, USA, Email: cynthia.rush@columbia.edu
\vspace{-0.9em}
\thanks{The work of M. Cardone was supported in part by the U.S. National Science Foundation under Grant CCF-1849757. 
}
}

\begin{document}
\IEEEoverridecommandlockouts
\maketitle

\begin{abstract}
We consider the problem of finding the subset of order statistics that contains the most information about a sample of random variables drawn independently from some known parametric distribution. We leverage information-theoretic quantities, such as entropy and mutual information, to quantify the level of informativeness and rigorously characterize the amount of information contained in any subset of the complete collection of order statistics. As an example, we show how these informativeness metrics can be evaluated for a sample of discrete Bernoulli and continuous Uniform random variables. Finally, we unveil how our most informative order statistics framework can be applied to image processing applications. Specifically, we investigate how the proposed measures can be used to choose the coefficients of the L-estimator filter to denoise an image corrupted by random noise. We show that both for discrete (e.g., salt-pepper noise) and continuous (e.g., mixed Gaussian noise) noise distributions, the proposed method is competitive with off-the-shelf filters, such as the median and the total variation filters, as well as with wavelet-based denoising methods.
\end{abstract}

\section{Introduction} 
\label{sec:Intro}
Consider a random sample  $X_1, X_2, \ldots, X_n$ drawn independently from some known parametric distribution $p(x| \theta)$ where the parameter $\theta$ may or may not be known.  Let  the random variables~(r.v.)  $X_{(1)} \le  X_{(2)} \le \ldots \le X_{(n)}$ represent the order statistics of the sample.  In particular, $X_{(1)}$ corresponds to the minimum value of the sample,  $X_{(n)}$ corresponds to the maximum value of the sample, and  $X_{( \frac{n}{2})}$   (provided that $n$ is even) corresponds to the median of the sample.   We denote the collection of the random samples as $X^n := (X_1, X_2,\ldots, X_n)$, and we use $[n]$ to denote the collection $\{1, 2, \ldots, n\}.$

As illustrated by comprehensive survey texts~\cite{HBS17, David2003Book}, order statistics have a broad range of applications including survival and reliability analysis, life testing, statistical quality control, filtering theory, signal processing, robustness and classification studies, radar target detection, and wireless communication.
In such a wide variety of  practical situations, some order statistics -- such as the minimum, the maximum, and the median -- have been analyzed and adopted more than others.
For instance, in the context of image processing (see also Section~\ref{sec:Appl}), a widely employed order statistic filter is the {\em median} filter.
However, to the best of our knowledge, there is not a {\em theoretical} study that justifies why certain order statistics should be preferred over others. Although such a universal\footnote{ A large body of the literature has focused on analyzing information measures of the (continuous or discrete) parent population of ordered statistics (examples include the differential entropy~\cite{baratpour2007some}, the R\'enyi entropy~\cite{baratpour2008characterizations,abbasnejad2010renyi}, the cumulative entropies~\cite{balakrishnan2020cumulative}, the Fisher information~\cite{zheng2009fisher}, and the $f$-divergence~\cite{OurArxiv}) and trying to show universal (i.e., distribution-free) properties for such information measures, see for instance~\cite{abbasnejad2010renyi,wong1990entropy,OurArxiv,ebrahimi2004information}.} choice can be justified when there is no knowledge of the underlying distribution, in scenarios where some knowledge is available a natural question arises: Can we somehow leverage such a knowledge to choose which is the ``best'' order statistic to consider?

The main goal of this paper is to answer the above question. Towards this end, we introduce and analyze a theoretical framework for performing `optimal' order statistic selection to fill the aforementioned
theoretical gap. 
Specifically, our framework allows to rigorously identify the subset of order statistics that contains the most information on a random sample.
As an application, we show how the developed framework can be used for image denoising to produce competitive approaches with off-the-shelf filters, as well as with wavelet-based denoising methods.
Similar ideas also have the potential to benefit other fields where order statistics find application, such as radar detection and classification.
With the goal of developing a theoretical framework for `optimal' order statistic selection,
in this work we are interested in answering  the following questions:   

\noindent \textbf{(1) } How much   `information' does a single  order statistic $X_{(i)}$ contain about the random sample $X^n$
for each $i \in [n]$?  We refer to the $X_{(i)}$  that contains the most information about the sample as \emph{the most informative order statistic}.

\noindent \textbf{(2) }  Let  $ \mathcal{S} \subseteq [n]$ be a set  of cardinality $|\mathcal{S}| = k$ and let $ X_{  ( \mathcal{S} )} =\{X_{(i)} \}_{i \in \mathcal{S}}$.  Which subset of order statistics $X_{  ( \mathcal{S} )} $ of size $k$ is the most informative with respect to the sample $X^n$?

\noindent \textbf{(3) } Given a set $\mathcal{S} \subseteq [n]$ and the collection of order statistics $X_{( \mathcal{S})}$, which \emph{additional} order statistic $X_{(i)}$ where $i \in [n]$ but $i \not \in \mathcal{S}$, adds the most information about the sample $X^n$?

One approach for defining the most informative order statistics, and the one that we investigate in this work, is to  consider the mutual information as a base measure of informativeness.  Recall that, intuitively, the mutual information between two variables $X$ and $Y$, denoted as $I(X; Y) = I(Y; X)$, measures the reduction in uncertainty about one of the variables given the knowledge of the other. Let $p(x,y)$ be the joint density of $(X,Y)$ and let $p(x), p(y)$ be the marginals. The mutual information is calculated as
\begin{equation}
\label{eq:MI_def}
I(X; Y) = \int \int p(x, y) \log \left(\frac{p(x, y)}{p(x)p(y)}\right) dx \, dy.
\end{equation}
The base of the logarithm determines the units of the  measure, and throughout the paper we use base $e$.  Notice that there is a relationship between the mutual information and the differential  entropy, namely,
\begin{equation}
I(X; Y) =h(Y) - h(Y|X),
\label{eq:MI_def_entropy}
\end{equation}
where the entropy and the conditional entropy are defined as
$h(Y) = - \int p(y) \log p(y) \ dy,$ and
$h(Y|X) = \int \int p(x, y)  \log( {p(x)}/{p(x, y)}) dy \ dx.$
The discrete analogue of~\eqref{eq:MI_def} replaces the integrals with sums, and~\eqref{eq:MI_def_entropy} holds with the differential entropy $h(Y)$ being replaced with its discrete version, denoted as $H(Y) = - \sum_y p(y) \log p(y)$.

In particular, if $X$ and $Y$ are independent -- so knowing one delivers no information about the other -- then the mutual information is zero.  
Differently,
if $X$ is a deterministic function of $Y$ and $Y$ is a deterministic function of $X$, then knowing one gives us complete information on the other.  If additionally, $X$ and $Y$ are discrete, the mutual information is then the same as the amount of information contained in $X$ or $Y$ alone, as measured by the entropy, $H(Y)$, since $H(Y|X) = 0$. If $X$ and $Y$ are continuous, the mutual information is infinite since $h(Y|X) = - \infty$ (because $(X, X)$ is singular with respect to the Lebesgue measure on $\mathbb{R}^2$).

 \section{Measures of Informativeness of Order Statistics} \label{sec:Measures}
 
In this section, we propose several metrics, all of which leverage the mutual information as a base measure of informativeness. We start by considering the mutual information between the sample $X^n$ and any order statistic $X_{(i)}$, i.e., $I(X_{(i)}; X^n)$ and find the index $i \in [n]$ that results in the largest mutual information. In the case of discrete r.v., we have
\begin{equation*}
\begin{split}
I(X_{(i)}; X^n) &= H(X^n) - H(X^n|X_{(i)}) \\
& = \sum_{x_{(i)}} \sum_{x^n} p(x_{(i)}, x^n) \log \left(\frac{p(x_{(i)}, x^n)}{p(x_{(i)})p(x^n)}\right).
\end{split}
\end{equation*}
Such an approach works only when the sample is composed of discrete r.v.\ and does not work for continuous r.v. The reason for this is that, as highlighted in Section~\ref{sec:Intro}, when $X^n$ is a collection of continuous r.v., then $I(X_{(i)}; X^n)=\infty$ as $h(X_{(i)}|X^n) = - \infty$.    

This idea of using mutual information, however, can be salvaged by introducing noise to the sample.  For example, the informativeness of  $X_{(i)}$ can be measured by considering $I(X_{(i)}; X^n+\sigma Z^n)$ where $Z^n := (Z_1, Z_2, \ldots, Z_n)$ is a vector of   i.i.d.\  Gaussian r.v.\  independent of $X^n$ with $ \sigma$ being the noise standard deviation. 
Next, based on the above discussion, we propose three potential measures of informativeness of order statistics about the sample $X^n$, all based on the mutual information measure. 

\begin{defin}\label{def:OrderstatMI} Let $Z^n := (Z_1, Z_2, \ldots, Z_n)$ be a vector of i.i.d.\  standard Gaussian r.v.\ independent of $X^n = (X_1, X_2, \ldots, X_n)$.   
Let $\mathcal{S}\subseteq [n]$ be defined as $$\mathcal{S} =\{ ( i_1,i_2,\ldots,  i_k ): 1 \le  i_1< i_2< \ldots <  i_k \le n  \},$$ with $|\mathcal{S}|=k$. We define the following three measures of order statistic informativeness: 
\begin{align}
\mathsf{r}_1(\mathcal{S},X^n)&=I(X^n; X_{  ( \mathcal{S} ) } ), \label{eq:r1Measure} \\
\mathsf{r}_2(\mathcal{S},X^n)&=   \lim_{\sigma \to \infty}  2 \sigma^2 I(X^n+\sigma Z^n; X_{  ( \mathcal{S} )} ),   \label{eq:r2Measure}\\
\mathsf{r}_3(\mathcal{S},X^n)&=   \lim_{\sigma \to \infty}  2 \sigma^2I(X^n; X_{  ( \mathcal{S} )}+ \sigma Z^k). \label{eq:r3Measure}
\end{align}  
\end{defin} 
In Definition~\ref{def:OrderstatMI}, the measure $\mathsf{r}_1(\mathcal{S},X^n)$ computes the mutual information between a subset of order statistics $X_{  ( \mathcal{S} )}$ and the sample $X^n$.   
The measure $  \mathsf{r}_2(\mathcal{S},X^n)$  computes the slope of the mutual information at $\sigma=\infty$: intuitively, as noise becomes large, only the most informative $X_{  ( \mathcal{S} )}$ should maintain the largest mutual information.
The measure  $  \mathsf{r}_3(\mathcal{S},X^n)$ is an alternative to  $  \mathsf{r}_2(\mathcal{S},X^n)$, with noise added to  $X_{  ( \mathcal{S} )}$ instead of  $X^n$. The limits in~\eqref{eq:r2Measure} and~\eqref{eq:r3Measure} always exist, but may be infinity.

One might also  consider similar measures as in~\eqref{eq:r2Measure} and~\eqref{eq:r3Measure}, but in the limit of $\sigma$ that goes to zero, namely
\begin{equation}
\begin{split}
\label{eq:r4r5Measure}
\mathsf{r}_4(\mathcal{S},X^n) &=   \lim_{\sigma \to 0} \frac{ I(X^n+\sigma Z^n;  X_{  ( \mathcal{S} )}   )}{ \frac{1}{2} \log(1 +\frac{1}{\sigma^2} )},  \\
\mathsf{r}_5(\mathcal{S},X^n) &=   \lim_{\sigma \to 0} \frac{ I(X^n; X_{  ( \mathcal{S} )} +\sigma Z^k)}{ \frac{1}{2} \log(1 +\frac{1}{\sigma^2} )}.   
\end{split}
\end{equation}
In particular, the intuition behind $\mathsf{r}_4(\mathcal{S},X^n)$ is that the most informative  set $X_{  ( \mathcal{S} )}$  should  have the largest increase in the mutual information as the  observed sample becomes less noisy.    The measure $\mathsf{r}_5(\mathcal{S},X^n)$ is an alternative to $\mathsf{r}_4(\mathcal{S},X^n)$  where the noise is added to  $X_{  ( \mathcal{S} )}$ instead of  $X^n$. 
However, as we prove next, these measures evaluate to
\begin{equation*}
\begin{split}
  \mathsf{r}_4(\mathcal{S},X^n) &= 0,   \text{ continuous and discrete r.v.,} \\
\mathsf{r}_5(\mathcal{S},X^n) &=    \left \{ \begin{array}{ll} k,   &  \text{ continuous r.v.,} \\
0, & \text{ discrete r.v..}
\end{array}  \right. 
\end{split}
 \end{equation*} 
 Hence, these are not useful measures of information.
 
\begin{proof}
To characterize $\mathsf{r}_4(\mathcal{S},X^n)$ in~\eqref{eq:r4r5Measure}, recall that
by the data processing inequality, if $X \to  Y \to Z$ is a Markov chain then $I(X;Z|Y) = 0$.  Now, since 
$X^n +\sigma Z^n \to X^n \to X_{(\mathcal{S})}$ is a Markov chain and $I(X^n+\sigma Z^n; X_{(\mathcal{S})}| X^n) = 0$, we therefore have that
$I(X^n+\sigma Z^n;  X_{(\mathcal{S})})  =   I(X^n+\sigma Z^n; X^n,X_{(\mathcal{S})}).$
Then, by the chain rule of the mutual information,
$I(X^n+\sigma Z^n; X^n,  X_{(\mathcal{S})} ) = I(X^n+\sigma Z^n; X^n)-  I(X^n+\sigma Z^n; X^n|  X_{(\mathcal{S})} )$,
and, 
\begin{align*}
&\mathsf{r}_4(\mathcal{S},X^n) = \lim_{\sigma \to 0} \frac{ I(X^n+\sigma Z^n;  X_{(\mathcal{S})}  )}{ \frac{1}{2} \log(1 +\frac{1}{\sigma^2} )} \notag\\
 &=  \lim_{\sigma \to 0} \frac{ I(X^n+\sigma Z^n; X^n)-  I(X^n+\sigma Z^n; X^n|  X_{(\mathcal{S})} )}{ \frac{1}{2} \log(1 +\frac{1}{\sigma^2} )} \nonumber \\
& =   \mathsf{d}(X^n)- \mathsf{d}(X^n| X_{(\mathcal{S})}),
\end{align*} 
where $\mathsf{d}(X^n)$ is known as the information dimension or R\'enyi dimension~\cite{guionnet2007classical,wu2012optimal}, namely
\begin{align}
\mathsf{d}(X^n)=\left \{ \begin{array}{ll} n   &  \text{ continuous r.v.} \\
0 & \text{ discrete r.v..}
\end{array}  \right.  
\label{eq:sfd_def}
\end{align}
Similarly, since $(X_{(\mathcal{S})} +\sigma Z^k) \to X_{(\mathcal{S})} \to  X^n $ is a Markov chain with $I(X_{(\mathcal{S})} +\sigma Z^k; X_{(\mathcal{S})} | X^n) = 0$, we obtain
\begin{align*}
\mathsf{r}_5(\mathcal{S},X^n) &=   \lim_{\sigma \to 0} \frac{ I(X^n; X_{(\mathcal{S})}+\sigma Z^k)}{ \frac{1}{2} \log(1 +\frac{1}{\sigma^2} )} \\
&=   \lim_{\sigma \to 0} \frac{ I( X_{(\mathcal{S})} ;  X_{(\mathcal{S})} +\sigma Z^k)}{ \frac{1}{2} \log(1 +\frac{1}{\sigma^2} )}  = \mathsf{d}(X_{(\mathcal{S})}),
\end{align*} 
where $\mathsf{d}(\cdot)$ is defined in \eqref{eq:sfd_def}.
\end{proof} 

\begin{remark}
We emphasize that the scaling and
 Gaussian noise used above were not chosen artificially.  It can be shown that  any
 absolutely continuous perturbation with a finite Fisher information would result in equivalent limits~\cite{guo2005additive}.  Therefore, the choice of Gaussian noise was simply made for the ease of exposition and the proof.
\end{remark} 

There are a few shortcomings of the  measures just introduced. For instance,  the elements of the most informative set are not ordered based on the amount of information that each element provides. Moreover,  at this point, we are unable to quantify the amount of information that an additional order statistic adds to a given collection $X_{( \mathcal{S})}$ of order statistics. 
These shortcomings can be remedied by considering a {\em conditional} version of the measures  introduced in Definition~\ref{def:OrderstatMI}. 
\begin{defin} \label{def:CondMeas}
Under the assumptions in Definition~\ref{def:OrderstatMI}, let $\mathcal{V} \subset [n]$ such that $\mathcal{S} \cap \mathcal{V} =\varnothing$. Then, we define three conditional measures of order statistic informativeness:  
\begin{align}
\mathsf{r}_1(\mathcal{S},X^n|  \mathcal{V} )&=I(X^n; X_{  ( \mathcal{S} ) } |   X_{  ( \mathcal{V} ) } ), \label{eq:r1MeasureC} \\
\mathsf{r}_2(\mathcal{S},X^n| \mathcal{V} )&=   \lim_{\sigma \to \infty}  2 \sigma^2 I(X^n+\sigma Z^n; X_{  ( \mathcal{S} )} |   X_{  ( \mathcal{V} ) }),   \label{eq:r2MeasureC}\\
\mathsf{r}_3(\mathcal{S},X^n| \mathcal{V} )&=   \lim_{\sigma \to \infty}  2 \sigma^2I(X^n; X_{  ( \mathcal{S} )}+ \sigma Z^k |   X_{  ( \mathcal{V} ) }).\label{eq:r3MeasureC}
\end{align}  
\end{defin}

\section{Characterization of the Informativeness Measures}
\label{sec:CharInfoMeas}
In this section,
we characterize the measures of informativeness of order statistics proposed in Definition~\ref{def:OrderstatMI} and Definition~\ref{def:CondMeas}. 
In particular, we have the following theorem.
\begin{thm} 
 \label{thm:RepresentationsCond}
Let $\mathcal{S} \subseteq [n]$ such that $|\mathcal{S}|=k$, and $\mathcal{V} \subset [n]$ such that $\mathcal{S} \cap \mathcal{V} =\varnothing$. Then, the metrics in Definition~\ref{def:CondMeas} evaluate to
\begin{align}
\mathsf{r}_1(\mathcal{S},X^n|  \mathcal{V} )&\!=\!  \left \{ \begin{array}{ll}  \!\!\!\!H(X_{(\mathcal{S})}  |X_{  ( \mathcal{V} ) }),   &  \!\text{for discrete r.v.,} \\
\!\!\!\!\infty,& \!\text{otherwise,}
\end{array}  \right.  \\
\mathsf{r}_2(\mathcal{S},X^n|  \mathcal{V} )&\!=\!   \E[ \| \E[X^n| X_{  ( \mathcal{V} ) }]\!-\! \E[X^n| X_{(\mathcal{S})},X_{  ( \mathcal{V} ) }] \|^2],
\label{eq:m2_solvedcond} \\
 \mathsf{r}_3(\mathcal{S},X^n|  \mathcal{V} )&=  \E[ \| X_{(\mathcal{S})} -\E[X_{(\mathcal{S})}|   X_{  ( \mathcal{V} ) } ] \|^2] . 
 \end{align}
Taking $ \mathcal{V}\! =\! \varnothing$ gives an evaluation of the metrics in Definition~\ref{def:OrderstatMI}, namely $\mathsf{r}_1(\mathcal{S},X^n)=  H(X_{(\mathcal{S})})$ for discrete r.v.\ and $\mathsf{r}_1(\mathcal{S},X^n)=  \infty$ otherwise, $\mathsf{r}_2(\mathcal{S},X^n) =    \E[ \| \E[X^n]- \E[X^n| X_{(\mathcal{S})}] \|^2]$, and $\mathsf{r}_3(\mathcal{S},X^n)=  \E[ \| X_{(\mathcal{S})} -\E[X_{(\mathcal{S})} ] \|^2]$.
\end{thm} 
\begin{proof}
For simplicity, we focus on the case $\mathcal{V}=\varnothing$. The proof for arbitrary $\mathcal{V}$ follows along the same lines. First, assume that $X^n$ is a sequence of discrete r.v. Then, by using the relationship between mutual information and entropy given in~\eqref{eq:MI_def_entropy} we have,
$
I(X^n; X_{(\mathcal{S})})= H(X_{(\mathcal{S})}) - H(X_{(\mathcal{S})}|X^n)=H(X_{(\mathcal{S})}),
$
where the last equality uses that $H(X_{(\mathcal{S})}|X^n) = 0$ since $X_{(\mathcal{S})}$ is fully determined given the value of the sequence $X^n$.  As mentioned in Section~\ref{sec:Measures}, if $X^n$ is a sequence of continuous r.v. then $I(X^n; X_{(\mathcal{S})})= h(X_{(\mathcal{S})}) - h(X_{(\mathcal{S})}|X^n) = \infty$ since $h(X_{(\mathcal{S})}|X^n) = -\infty$. This characterizes $\mathsf{r}_1(\mathcal{S},X^n)$

We now characterize the measure $\mathsf{r}_2(\mathcal{S},X^n)$. We have that
\begin{align}
& \mathsf{r}_2(\mathcal{S},X^n)   =2 \lim_{\sigma \to \infty}   \sigma^2 I(X^n+\sigma Z^n; X_{(\mathcal{S})}) \nonumber \\
   & \stackrel{{\rm{(a)}}}{=}   2 \lim_{\snr \to 0}  \frac{I( \sqrt{\snr} X^n+ Z^n; X_{(\mathcal{S})})}{\snr} \nonumber \\
   & \stackrel{{\rm{(b)}}}{=} 2  \frac{\rm d}{{\rm d} \snr }  I( \sqrt{\snr} X^n+ Z^n; X_{(\mathcal{S})})\Big \lvert_{\snr=0}  \nonumber \\
   &\stackrel{{\rm{(c)}}}{=}   \E\left[ \| X^n -\E[X^n|  Z^n] \|^2- \| X^n -\E[X^n|  Z^n, X_{(\mathcal{S})}] \|^2\right] \nonumber \\
   & \stackrel{{\rm{(d)}}}{=} \E\left[ \| X^n -\E[X^n] \|^2\right] \!-\!  \E\left[\| X^n -\E[X^n| X_{(\mathcal{S})}] \|^2\right], \label{eq:r2_s5}
\end{align} 
where the labeled equalities follow from:
$\rm{(a)}$ defining $\snr = 1/\sigma^2$ and noting that $I(aX;Y) = I(X; Y)$ for a constant $a$;
$\rm{(b)}$ using the fact that
\[ \lim_{\snr \to 0}  \frac{f(\snr) - f(0)}{\snr} =  \frac{\rm d}{{\rm d} a } f(a) \Big \lvert_{a=0},\]
where $f(a) = I( \sqrt{a} X^n+ Z^n; X_{(\mathcal{S})})$ with $f(0) = I(Z^n; X_{(\mathcal{S})}) = 0$;
$\rm{(c)}$ using the generalized I-MMSE relationship \cite[Thm.~10]{I-MMSE}
since $X_{(\mathcal{S})} \to X^n \to (\sqrt{\snr} X^n+ Z^n)$ is a Markov chain;
and $\rm{(d)}$ since $Z^n$ is independent of $X^n$.

To conclude the proof of $ \mathsf{r}_2(\mathcal{S},X^n)$ in~\eqref{eq:m2_solvedcond}, we would like to show that~\eqref{eq:r2_s5} is equal to $ \E[ \| \E[X^n]- \E[X^n| X_{(\mathcal{S})}] \|^2]$. We start by noting that
\begin{align}
&\E \left[ \| \E[X^n]- \E[X^n| X_{(\mathcal{S})}] \|^2 \right] \notag \\
&=  \E\left[ \left\| ( \E[X^n] - X^n ) +(X^n - \E[X^n| X_{(\mathcal{S})}]) \right\|^2 \right] \nonumber
\\& =  \E \left[\|  \E[X^n] - X^n \|^2 \right] + \E\left [ \|X^n - \E[X^n| X_{(\mathcal{S})} ] \|^2 \right] \nonumber  \\
& \hspace{8mm}+2\E \left[ ( \E[X^n] - X^n ) ^T (X^n - \E[X^n| X_{(\mathcal{S})} ]) \right]. \label{eq:AuxiRsr2}
 \end{align}
Moreover, we note that
\begin{align}
&-2\E \left[ ( \E[X^n] - X^n)^T (X^n - \E[X^n| X_{(\mathcal{S})}]) \right]\notag
\\&  = 2\E \left[   (X^n-  \E[X^n])^T X^n  - (X^n - \E[X^n]) ^T  \E[X^n| X_{(\mathcal{S})}]\right]  \notag \\
& \stackrel{{\rm{(a)}}}{=} 2\E \left[ (X^n) ^T (X^n  -   \E[X^n| X_{(\mathcal{S})}]) \right]  \notag \\
& \stackrel{{\rm{(b)}}}{=} 2  \E \left[\| X^n - \E[X^n| X_{(\mathcal{S})}]\|^2 \right], \label{eq:Auxi2Rsr2}
 \end{align}
where the labeled equalities follow from:
$\rm{(a)}$ the fact that 
\begin{align*}
&\E  \left [(\E[X^n|)^T(X^n - \E[X^n| X_{(\mathcal{S})}]) \right] \\
&=  (\E[X^n|)^T \E \left [ X^n - \E[X^n| X_{(\mathcal{S})}] \right]
\\& =(\E[X^n|)^T   \left( \E[X^n] - \E[\E[X^n| X_{(\mathcal{S})}]] \right)
\\& =(\E[X^n|)^T  \left ( \E[X^n] - \E[X^n]  \right) = 0,
\end{align*}
where in the third equality we have used the law of total expectation; 
and $\rm{(b)}$ using the orthogonality principle~\cite{Kay97}, which states that $\E [ (  \E[X^n| X_{(\mathcal{S})}]) ^T (X^n  -   \E[X^n| X_{(\mathcal{S})}])]  =0.$

By substituting~\eqref{eq:Auxi2Rsr2} back into~\eqref{eq:AuxiRsr2}, we obtain
\begin{align*}
&\E\left[ \| \E[X^n]- \E[X^n| X_{(\mathcal{S})}] \|^2 \right] \\
&= \E\left [ \|  \E[X^n] - X^n\|^2 \right] - \E\left [\|X^n - \E[X^n| X_{(\mathcal{S})}] \|^2 \right],
\end{align*}
which is precisely~\eqref{eq:r2_s5}. Hence, 
$
\mathsf{r}_2(\mathcal{S},X^n)  = \E[ \| \E[X^n]- \E[X^n| X_{(\mathcal{S})}] \|^2 ].
$

We now characterize $\mathsf{r}_3(\mathcal{S},X^n)$.  It follows by the data processing inequality, that $I(X; Z) = I(X; Y)$ for a Markov chain $X \to Y \to Z$ if $I(X; Y | Z) = 0$.  Notice that in our problem, $(X_{(\mathcal{S})} +\sigma Z^k) \to X_{(\mathcal{S})} \to  X^n $ forms a Markov chain with $I(X_{(\mathcal{S})} +\sigma Z^k; X_{(\mathcal{S})} | X^n) = 0$. Thus,
$
I(X_{(\mathcal{S})} +\sigma Z^k; X^n) = I(X_{(\mathcal{S})} +\sigma Z^k; X_{(\mathcal{S})}).
$
Therefore,
\begin{align*}
\mathsf{r}_3(\mathcal{S},X^n) & =   \lim_{\sigma \to \infty}  2 \sigma^2 I(X^n; X_{(\mathcal{S})}+ \sigma Z^k) \nonumber \\
&=   \lim_{\sigma \to \infty}  2 \sigma^2 I(X_{(\mathcal{S})}; X_{(\mathcal{S})}+ \sigma Z^k) \nonumber \\
&=   \E\left[ \| X_{(\mathcal{S})} -\E[X_{(\mathcal{S})} ]  |^2 \right],  
\end{align*} 
where the last limit is a standard result and can for example be found in~\cite[Corollary~2]{prelov2004second}.  %
\end{proof} 
By leveraging 
Theorem~\ref{thm:RepresentationsCond}, 
we can now construct procedures that answer the three questions raised in Section~\ref{sec:Intro}.   
Specifically, given $m \in [3]$, we propose the following  three approaches: 

\smallskip
\noindent
\textbf{(1) }  \emph{Marginal Approach}:  Generate one set of cardinality $k$  according to
  \begin{align}
\label{eq:Marg}
 \bar{\mathcal{S}}_{m}^M= \{ (i_1,\ldots, i_k): \,  & r_m( i_1,X^n ) \ge \ldots \ge   r_m( i_k,X^n ) ,  \nonumber  \\
 & 1 \le  i_1< \ldots <  i_k \le n    \}.
 \end{align}
This  approach generates an {\em ordered} set $\bar{\mathcal{S}}_{m}^M$
of indices of order statistics, listed from the (first) most informative to the $k$-th most informative, 
and  quantifies the amount of information that an individual order statistic contains about the sample.  

\smallskip
\noindent
\textbf{(2) } \emph{Joint Approach}: Generate one set of cardinality $k$ with 
\begin{equation}
\label{eq:Joint}
\bar{\mathcal{S}}_m^J \in \argmax_{\mathcal{S} \subseteq [n], \,  |\mathcal{S}|=k} \mathsf{r}_m(\mathcal{S},X^n).
\end{equation} 
Now $\bar{\mathcal{S}}_m^J$ contains the indices of the $k$ order statistics that are the most 
informative
about the sample.

\smallskip
\noindent
\textbf{(3) }\emph{Sequential Approach}: Generate one set of cardinality $k$ according to
\begin{align} 
\label{eq:Cond}
 \bar{\mathcal{S}}_{m}^S =\{   &(i_1,\ldots, i_k): \,  \nonumber  \\
 & r_m( i_t,X^n | \mathcal{V}_{t-1} )  \ge \max_{ j \in [n]: j \notin  \mathcal{V}_{t-1} }    r_m( j ,X^n | \mathcal{V}_{t-1} ), 
\nonumber \\
&  \mathcal{V}_{t}  = (i_1,\ldots, i_t),  t\in   [k] ,  \mathcal{V}_{0}   = \varnothing  \}. 
\end{align} 
This approach produces an ordered set, $\bar{\mathcal{S}}_{m}^S$,
of indices of order statistics where $i_t$ is the most 
informative order statistic given that the information of $t-1$ order statistics has already been incorporated (captured by the conditioning term).

In the next section we show that the sets  $ \bar{\mathcal{S}}_{m}^M, \bar{\mathcal{S}}_m^J$ and $ \bar{\mathcal{S}}_{m}^S$ may not be the same, even in simple cases. Thus, 
 the application of interest and target analysis should guide the choice of which approach to use (i.e., which of the three questions raised in Section~\ref{sec:Intro}  is most relevant for the problem at hand).

\section{Evaluation of the Informativeness Measures}
\subsection{Discrete Random Variables: The Bernoulli Case} 
We assess the three measures in Theorem~\ref{thm:RepresentationsCond} for the case of a sample of discrete r.v.\ in Lemma~\ref{lemma:Binary} (proof in Appendix~\ref{app:BernEx}). In particular, Lemma~\ref{lemma:Binary}  studies the Bernoulli case, and in Section~\ref{sec:Appl} we consider another discrete distribution with applications to  image processing.
The results presented here rely heavily on Lemma~\ref{lem:OS_dist} in Appendix~\ref{app:Sec:ProofOfJointPMFD} to compute the joint distribution of $k$ order statistics. 
\begin{lem}   \label{lemma:Binary}
Let $X^n$ be sampled as i.i.d.\ Bernoulli with success probability $p$. 
Let $B$ be a Binomial$(n, 1-p)$ r.v.\ and $B'$ be a Binomial$(n-1, 1-p)$ r.v.
Then,
\begin{align}
\mathsf{r}_1(i,X^n)&=h_b(P(B <  i)), \label{eq:r1Bern}\\
\mathsf{r}_2(i,X^n)&= \frac{np^2}{P(B<i)}\Big[ P(B' < i)\Big]^2 \notag \\
&\qquad +  \frac{np^2}{P(B \ge i)} \Big[P(B' \geq i) \Big]^2 - np^2, \label{eq:r2Bern} \\
 \mathsf{r}_3(i,X^n) & =P(B <  i)P(B \geq  i)   \label{eq:r3Bern},
\end{align} 
where $h_b(t) := -t \log(t) -(1-t)\log(1-t)$ is the binary entropy function.  
\end{lem} 

\begin{remark}
\label{rem:BernEx}
Consider $x(1-x)$ for $x \in (0,1)$, which is symmetric and convex with the maximum occurring at $x=1/2$.  
Thus, $\mathsf{r}_3(i,X^n)$ in~\eqref{eq:r3Bern} is maximized by the $i$ such that $P(B \ge i)$ or  $1 - P(B \ge i)$ is as close to $1/2$ as possible. Hence, the maximizer $i$  is a median of $B$, namely,
\begin{equation}
\begin{split}
&i^\star_3(X^n)  = \argmax_{i \in [n]} \mathsf{r}_3(i,X^n)  \\
&=  \argmin_{i \in \{\lfloor  n (1-p) \rfloor, \lceil n (1-p) \rceil\}} \Big\{\Big|P(B \ge  i) - \frac{1}{2}\Big|\Big\}. \label{eq:istar3_bernoulli}
\end{split}
\end{equation}  
Moreover, since the binary entropy function $h_b(t)$ is increasing on $0 < t \le 1/2$ and decreasing on $1/2 \leq t < 1$, the maximizer for $\mathsf{r}_1(i,X^n)$ in~\eqref{eq:r1Bern} will also be given by \eqref{eq:istar3_bernoulli}.
\end{remark} 
When $X^n$ is sampled i.i.d.\ Bernoulli with  probability $p$, the `information' in the order statistics $0 \leq X_{(1)} \leq X_{(2)} \leq \ldots \leq X_{(n)} \leq 1$ is simply the counts of $0$'s and $1$'s present in the data. In terms of the order statistics, the `information' lies in the location of the switch point (if there is one), i.e., the $i$ where $X_{(i)} = 0$ but $X_{(i+1)} = 1$.  Since we expect $\mathbb{E}[X^n] = np$ of the samples to take the value $1$,  the switch point is expected to occur at $\text{round}(n(1-p))$, and Remark~\ref{rem:BernEx} (at least for $\mathsf{r}_1(\cdot,\cdot)$ and $\mathsf{r}_3(\cdot,\cdot)$) tells us that the `most informative' order statistic is where we expect the switch point to occur.    
In the next proposition, 
we  further show that, as the sample size grows, the most informative order statistic significantly dominates the  other statistics for measures~\eqref{eq:r1Bern} and~\eqref{eq:r3Bern}.
\begin{prop} 
\label{prop:LimLargeN}
Let $X^n$ be  i.i.d.\ Bernoulli with success probability $p \in (0,1)$.  For  any $c\in (0,1)$ independent of $n$, we~obtain
\begin{align*}
\lim_{n \to \infty}   \mathsf{r}_1( \lfloor c n \rfloor, X^n) &=  \left \{ \begin{array}{ll}    \log(2),  &  c=(1-p) ,\\
0 , & \text{otherwise} ,
\end{array}  \right. \\
\lim_{n \to \infty}   \mathsf{r}_3( \lfloor c n \rfloor, X^n)&=  \left \{ \begin{array}{ll}   1/4,  &  c=(1-p) ,\\
0 , & \text{otherwise} .
\end{array}  \right. 
\end{align*}
The same result holds  when $\lfloor \cdot \rfloor$ is replaced by $\lceil \cdot \rceil$.
\end{prop} 
\begin{proof}
From de Moivre-Laplace theorem \cite{feller2008introduction}, we know that for $B \sim \text{Binomial}(n, 1-p)$ the distribution of $ \frac{B-n(1-p)}{ \sqrt{n p(1-p)}}$ converges to the standard normal distribution. Hence,
\begin{align*}
& \lim_{n \to \infty} P( B < \lfloor c n \rfloor  ) \notag \\
 &=   \lim_{n \to \infty} P \left(  \frac{B-n(1-p)}{ \sqrt{n p(1-p)}} <  \frac{\lfloor c n \rfloor - n(1-p) }{ \sqrt{n p(1-p)}}  \right ) \notag\\
 &=   \lim_{n \to \infty} \Phi \left( \frac{\lfloor c n \rfloor - n(1-p) }{ \sqrt{n p(1-p)}}  \right ) =  \left \{\begin{array}{cc} 
 1,  &  c>(1-p),\\ 
 1/2, & c=(1-p),\\
 0 , &  c<(1-p),\\ 
  \end{array} \right. 
\end{align*} 
where  $\Phi(\cdot)$ is the cumulative distribution function of the standard normal. 
Inserting the limit  above into the expressions for $\mathsf{r}_1(\cdot,\cdot)$ and $\mathsf{r}_3(\cdot,\cdot)$ in~\eqref{eq:r1Bern} and~\eqref{eq:r3Bern} completes the proof.
\end{proof}
In the above, we focused our analysis on
the \emph{single} most informative order statistic.
We now want to consider \emph{sets} 
$\bar{\mathcal{S}}_{m}^M$, $\bar{\mathcal{S}}_m^J$ and $ \bar{\mathcal{S}}_{m}^S$ defined in~\eqref{eq:Marg}--\eqref{eq:Cond}.
For simplicity, we consider measure $\mathsf{r}_1(\cdot,\cdot)$ and an i.i.d.\ Bernoulli sample of size $n=19$ with $p=0.5$.  
Then  for set sizes $k \in [4]$  
we find
\begin{align*}
&\bar{\mathcal{S}}_{1}^M \rightarrow  \{10\}, \{10, 9 \},  \{10, 9, 11 \}, \{10, 9, 11, 8 \}; \\
&  \bar{\mathcal{S}}_{1}^J \rightarrow  \{10\}, \{9, 11 \},  \{10, 8, 12 \}, \{10,8,12,9\};\\
 &\bar{\mathcal{S}}_{1}^S  \rightarrow \{10\},  \{10,8\}, \{10,8,12\}, \{10,8,12,9\}.
\end{align*} 
Notice that the three sets can  all be different (e.g., when $k=2$) and we find that this difference 
becomes more drastic when any of the following occurs: $n$ increases, the size of the r.v.\ support increases, or the distribution becomes more asymmetric. 
To interpret the above, consider only the $k=2$ collection. From $\bar{\mathcal{S}}_{1}^M$, we know that the $10^{th}$ statistic is the most informative and the $9^{th}$ is the second most. However, the pair of most informative statistics is the $9^{th}$ and $11^{th}$ by $\bar{\mathcal{S}}_{1}^J$. From $\bar{\mathcal{S}}_{1}^S$, we know that, given the most informative (the $10^{th}$), the $8^{th}$  provides the most additional information.

\subsection{Continuous Random Variables: The Uniform Case} 

Now we look at an example for a sample of continuous random variables in Lemma~\ref{example:uniformContinious} (proof in Appendix~\ref{app:UnifExamp}).  
Remember that, from Theorem~\ref{thm:RepresentationsCond}, we have that the metric $\mathsf{r}_1(\cdot,\cdot)$ is infinity for continuous r.v., and hence we here focus on $\mathsf{r}_2(\cdot,\cdot)$ and $\mathsf{r}_3(\cdot,\cdot)$.
In particular, Lemma~\ref{example:uniformContinious} studies a Uniform sample, and in Section~\ref{sec:Appl} we consider another continuous distribution with applications to  image processing.
Throughout this section we use Lemma~\ref{lem:OS_cont}, in Appendix~\ref{app:Sec:ProofOfJointPMFC}, to compute the joint distribution of $k$ order statistics.

\begin{lem}\label{example:uniformContinious} Let $X^n$ be sampled as i.i.d.\ $\mathcal{U}(0,a)$ for $a>0$, i.e., sampled i.i.d.\ uniform on the interval $(0,a)$  and, for $k \in \{2,3\}$ define
$i^\star_k(X^n) = \argmax_{i \in [n]} \mathsf{r}_k(i,X^n).$
Then,
\begin{subequations}
\label{eq:ExContUn}
\begin{equation}
  \mathsf{r}_2(i,X^n)=\frac{ a^2 i (n+1-i)}{4n(n+2)},
\end{equation}
\begin{equation}
   \mathsf{r}_3(i,X^n)=  \frac{ a^2  i (n+1-i)}{ (n+1)^2 (n+2)}, \qquad \text{and}
\end{equation}
\begin{equation}
i^\star_2(X^n) = i^\star_3(X^n) \in  \left\{ \left \lceil \frac{n+1}{2} \right\rceil  ,\left \lfloor \frac{n+1}{2} \right\rfloor   \right  \}. \label{eq:i3starUn}
\end{equation}
\end{subequations}
\end{lem}
\begin{remark}
Lemma~\ref{example:uniformContinious} also encompasses the case where $X^n$ is sampled as i.i.d.\ $\mathcal{U}(a,b)$ for general $ a < b$ since the mutual information, which characterizes $\mathsf{r}_2(i,X^n)$ and $\mathsf{r}_3(i,X^n)$ (see Definition~\ref{def:OrderstatMI}), has the property that $I(X+c;Y) = I(X;Y)$ when $c$ is some constant.
\end{remark}
\begin{remark}
For $c\in(0,1)$ independent of $n$, metrics  $\mathsf{r}_2 (\cdot, \cdot)$ and $\mathsf{r}_3(\cdot, \cdot)$ have the following  behaviors as $n$ goes to infinity:
\begin{align*}
\lim_{n \to \infty} \mathsf{r}_2( \lfloor cn \rfloor, X^n)&=  {a^2c(1-c)}/{4}, \\
\lim_{n \to \infty}  n \cdot \mathsf{r}_3( \lfloor cn \rfloor, X^n)&= a^2 c(1-c).  
\end{align*} 
\end{remark}
We conclude this section by again considering the sets of most informative order statistics 
$\bar{\mathcal{S}}_{m}^M$, $\bar{\mathcal{S}}_m^J$ and $ \bar{\mathcal{S}}_{m}^S$ in~\eqref{eq:Marg}--\eqref{eq:Cond}. Specifically, for an i.i.d.\ sample uniform on $(0,a)$ with $a=1$ and $n=5$, the sets of sizes $k \in [4]$ are given by 
\begin{align*}
&\bar{\mathcal{S}}_{3}^M \rightarrow   \{3 \}, \{3,2\}, \{3,2,4\}, \{ 3,2,4,1\}; \\
& \bar{\mathcal{S}}_{3}^J \rightarrow    \{ 3 \}, \{3,2\}  , \{3,2,4\},\{ 3,2,4,1\} ;\\
&\bar{\mathcal{S}}_{3}^S  \rightarrow  \{3 \}, \{3,5\},  \{3,5,1 \},  \{   3,     5,     1,     4 \}.
\end{align*} 
Similarly to the discrete case, we see that it is possible for the  approaches to result in different sets.
To interpret the above, consider only the $k=2$ collection. From $\bar{\mathcal{S}}_{3}^M$, we know that the $3^{rd}$ statistic (the median) is the most informative and the $2^{nd}$ is the second most. By $\bar{\mathcal{S}}_{3}^J$, the same order statistics form the most informative pair. However, from $\bar{\mathcal{S}}_{3}^S$, we know that given the most informative (the $3^{rd}$), the $5^{th}$  provides the most additional information.

\section{Applications} 
\label{sec:Appl}

In this section, we show how the informativeness framework for order statistics just developed can be used in image processing applications.
We begin by reviewing some of the details about order statistics filters, which represent a class of non-linear filters.

\subsection{Order Statistics Filtering} 
Consider the following discrete-time filter, referred to as an {\bf L-estimator} in the remainder of the~paper.
\begin{defin} \label{def:Lestimator}
Define a filter
\begin{equation}
\label{eq:L-filter} 
Y_t =  \sum_{k=1}^n \alpha_k  X_{(k)},  \quad t  \in \mathbb{Z},
\end{equation}
where: (i) $X_{(k)}$, for $k \in [n],$ is the $k$-th order statistic of an  i.i.d.\ sequence $X_{t+i-1},$ for $i \in [n+1]$; 
(ii) $n$ is the filtering window width; and (iii) $\alpha_k \geq 0$'s, for $k \in [n],$ are the coefficients of the filter
such that $\sum_{k=1}^n \alpha_k=1$.
This filter is known as an \emph{L-estimator} in robust statistics~\cite{huber2004robust} and as an \emph{order statistics filter} in  image processing~\cite{pitas1992order, bovik1983generalization}. 
\end{defin}

 The general form of the L-estimator encompasses a large number of linear and non-linear filters. Examples are: 
\begin{enumerate}
\item  \emph{moving-average filter:} $\alpha_k=1/n,$ for all $k \in [n]$;  
\item \emph{median filter:} (by considering odd values of $n$) $\alpha_k=1$ for $k= (n+1)/2$ and $\alpha_k=0$ for $k\neq (n+1)/2$;  
\item \emph{maximum filter:} $\alpha_n=1$ and $\alpha_k=0$ for $k\neq n$; 
\item \emph{minimum filter:} $\alpha_1=1$ and $\alpha_k=0$ for $k\neq 1$;  
\item \emph{midpoint filter:} $\alpha_1 = \alpha_n =1/2$ and $\alpha_k=0$ for $k\neq 1,n$;
\item  \emph{$r$-th ranked-order filter:} $\alpha_r=1$ and $\alpha_k=0$ for $k\neq r$. 
\end{enumerate}
The L-estimator in Definition~\ref{def:Lestimator} has been extensively studied in the literature~\cite{viswanathan199823,yang2011order,HBS17,chernoff1967asymptotic,hosking19988,lloyd1952least,blom1962nearly,tukey1974nonlinear}. A comprehensive survey of their applications and, more generally, of order statistics is given in~\cite{HBS17}.
It is important to highlight that the L-estimator in~\eqref{eq:L-filter} forms a restricted class of estimators, and, as such, it is possible that other estimators, like the maximum likelihood, may have  better efficiency.  Nonetheless, it was shown in~\cite{chernoff1967asymptotic} that for a certain choice of weights, the estimator in~\eqref{eq:L-filter} attains the Cram\'er-Rao bound asymptotically and, hence, is asymptotically  efficient.  For an excellent survey on L-estimators, the interested reader is referred to~\cite{hosking19988}.  

The optimal choice of the coefficients in~\eqref{eq:L-filter} has received considerable attention in the context of scale-and-shift models.  Specifically, suppose that  the $X_i$'s are generated i.i.d.\ according to a cumulative distribution function, $F( \frac{x-\lambda}{\sigma})$,  where the location parameter, $\lambda$, and the scaling parameter, $\sigma$, are
unknown. The best unbiased estimator of $(\lambda, \sigma)$  under the mean squared error (MSE) criterion was found in~\cite{lloyd1952least}. This approach, however, requires  computation and inversion of covariance matrices of order statistics and is often prohibitive.  To overcome this, the authors of~\cite{blom1962nearly} proposed a choice of coefficients resulting in an approximately minimum variance, while depending only on $F(\cdot)$ and the probability density function (pdf), and only requiring inversion of a $2 \times 2$ matrix.

Our interest in this work lies in applications of order statistics  to image processing, where 
the median filter is the most popular choice~\cite{tukey1974nonlinear}. 
The work in~\cite{bovik1983generalization} also applies the L-estimator to image processing in a setting where the image is assumed to be corrupted by additive noise and the optimal MSE estimator of~\cite{lloyd1952least} was used. 
A comprehensive survey of applications of order statistics to digital image processing can be found in~\cite{pitas1992order}.    
 
For image processing, using a parametric scale-and-shift model might be too simplistic as it only models additive noise and a variety of widely-used image processing noise models, such as salt and pepper or speckle noise, cannot be modeled as additive.  
Moreover, the majority of the distortions encountered in practice are  discrete in nature, and hence one needs to work with discrete, instead of continuous, order statistics. 
Another issue that arises with the aforementioned approaches to choosing the optimal coefficients in~\eqref{eq:L-filter} is the use of the MSE as the fidelity criterion.  
Indeed, it turns out that the MSE is not a good approximation of the human perception of image fidelity~\cite{wang2009mean,pappas2000perceptual}.  
Thus, coefficients that are optimal for the MSE might not be the best choice if the goal is to optimize the human perceptual criterion for image quality.

We will use the measures in Section~\ref{sec:CharInfoMeas} to choose the L-estimator coefficients. This approach benefits from the fact that it can be applied to both continuous and discrete models. Moreover,
reliance on the MSE 
can be avoided, and signal fidelity can instead be measured using alternative 
quantities like the entropy.  
Our goal is to show that  selecting the L-estimator coefficients
using the most informative order statistics is a viable and competitive approach, worth further exploration. 
We compare the performance of the proposed L-estimator to that of several state-of-the-art denoising methods such as the total variation filter ~\cite{rudin1992nonlinear},
 and three different implementations of the wavelet-based filters  namely empirical Bayes \cite{johnstone2004needles}, Stein's Unbiased Estimate of Risk (SURE) \cite{donoho1994ideal} and   False Discovery Rate (FDR) \cite{pizurica2006review}.  

 Our simulations use the image in Fig.~\ref{fig:OrFig}, which has
  \begin{figure}
\centering
\fbox{\includegraphics[width=6cm]{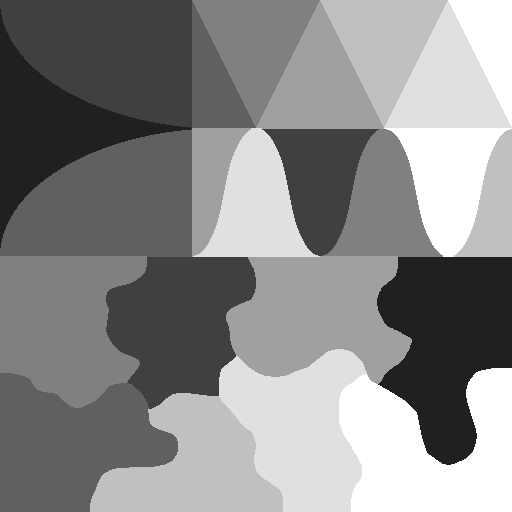}}
\caption{Test image.}
\label{fig:OrFig}
\end{figure} 
$N=(512)^2$ pixels.
As there is no universally-used performance metric for image reconstruction,  we consider several well-known ones: 
(i) the MSE normalized by $N$;
(ii) the peak signal-to-noise ratio (PSNR), measured in dB;
(iii) the structural similarity (SSIM) index~\cite{wang2004image}, taking values between $0$ and $1$ where $1$ is  perfect reconstruction;
and (iv) the image quality index IQI~\cite{wang2002universal}, taking values between $-1$ and $1$ where $1$ is perfect reconstruction.

\subsection{Image Denoising in Salt and Pepper Noise} \label{sec:SPNOISE}
We analyze gray scale image denoising where  pixels are typically $8$-bit data values ranging from $0$ (black) to $255$ (white). We use an observation model where an unknown pixel $x \in [0:255]$ is corrupted by the salt and pepper noise. Let  $ P(x \text{ is corrupted by pepper noise}| x \text{ is noisy} )= \rho_1$ and $P(x \text{ is noisy})=\rho$. We model the noisy observation $X$ with a probability mass function (pmf):
\begin{subequations}
\label{eq:NoiseModel}
\begin{align}
&P(x \text{ corrupted by pepper noise})\!=\!P(X\!=\! 0)\!=\!\!  \rho_1  \rho,\\
&P(x \text{ noise-free})= P(X= x)=  1- \rho,\\
&P(x \text{ corrupted by salt noise})\!=\!P(X\!=\! 255)\!=\!  (1\!-\!\rho_1) \rho.
\end{align} 
\end{subequations}
In the above, $\rho$  corresponds to the percentage of  pixels corrupted by noise, and $\rho_1$ is the percentage of pixels  corrupted by pepper noise.  

The pseudocode in Algorithm~\ref{algo:Deno} summarizes our general image denoising algorithm based on the L-estimator. 
In particular, we use a square-shaped  window of size $w \times w$ to sample the pixels of an image.
Moreover, if $\rho_1$ and $\rho$ are unknown, their estimates can be computed as 
\begin{equation}
\begin{split}
\hat{\rho}_1 &=\frac{  {\sum_{t=1}^N} \mathsf{1}_{\{X_t=0\} }   }{  {\sum_{t=1}^N} (\mathsf{1}_{\{X_t=0\}}   + \mathsf{1}_{\{X_t=255\} }  )}\\ 
\hat{\rho} &=   \frac{1}{N} {\sum_{t=1}^N} (\mathsf{1}_{\{X_t=0 \}}   + \mathsf{1}_{\{X_t=255 }\}  ),  \label{eq:EstProb}
\end{split}
\end{equation}
where $\mathsf{1}_{\{ \cdot \}}$ is the indicator function, and $N$ is the  number of pixels in the image.  
The estimators  
in~\eqref{eq:EstProb} perform well
if the original image contains very few pixel values exactly equal to $0$ and $255$, but since these are the extremes of possible pixel values, this is often reasonable to assume. 

Choosing $\mathsf{r}_1(\cdot,\cdot)$ as the performance metric offers several benefits.
First, the received  data $X^n$ for $n=w^2$ is discrete, and hence entropy is a natural choice for informativeness measure. 
Second,  the measures $\mathsf{r}_2(\cdot,\cdot)$ and $\mathsf{r}_3(\cdot,\cdot)$ depend on the values of the support of $X^n$. Thus, one would need to specify the value of the unknown parameter $x$ in~\eqref{eq:NoiseModel}. 
In contrast, the measure $\mathsf{r}_1(\cdot,\cdot)$ does not depend on the support values but only on the relative positions of the support points. Hence, the parameter $x$ can be left unspecified, and we only assume that it lies in the range $[0:255]$.

\begin{algorithm}
 \caption{Image denoising  based on the L-estimator.}
\label{algo:Deno}
  \begin{algorithmic}[1]
  	\renewcommand{\algorithmicrequire}{\textbf{Input:}}
  	\renewcommand{\algorithmicensure}{\textbf{Output:}}
  	\Require Image; Size $w$ of the square-shaped window; Probabilities $\rho_1$ and $\rho$ or their estimates in~\eqref{eq:EstProb}.
  	\Ensure Reconstructed image.
\State Set the length of the sequence $n=w^2$. Sample the square window of size $w \times w$ and collect the samples in a vector of length $n$. This constitutes the noisy sequence $X^{n} = \{X_i, \text{for } i \in [n]\}$. 
  	\State Compute $\mathsf{r}_1   (k,X^n   )$ in Theorem~\ref{thm:RepresentationsCond} for all $k \in [n]$ by using~\eqref{eq:NoiseModel}.
\State Compute the coefficients $\alpha_k$'s for all $k \in [n]$ for the L-estimator in Definition~\ref{def:Lestimator} as follows:
\begin{equation}
\begin{split}
&\text{If }  \rho < 0.5, \text{ assign } \alpha_k=  \frac{ \mathsf{r}_1^{-1} \left (k,X^{n} \right )}{  \sum_{i=1}^{n} \mathsf{r}_1^{-1}\left (i,X^{n} \right )  },  \\
& \text{ otherwise, assign } \alpha_k=  \frac{ \mathsf{r}_1(k,X^{n})}{  \sum_{i=1}^{n} \mathsf{r}_1(i,X^{n})  }. 
\label{eq:FilterCoefficientsSequentail} 
\end{split}
\end{equation} 
\State Apply the L-estimator in Definition~\ref{def:Lestimator} to the samples in Step~1 with the coefficients in~\eqref{eq:FilterCoefficientsSequentail}.
  \end{algorithmic}

\end{algorithm}

We now explain our choice of the coefficients in~\eqref{eq:FilterCoefficientsSequentail} for the low-noise regime, i.e., $\rho <0.5$, and for the high-noise regime, i.e., $\rho \geq 0.5$. 
We start by noting that in the ordered sample $X_{(1)},\ldots, X_{(n)}$ with $n=w^2$, \emph{approximately}: 
(i) the  first $\rho_1 \rho n$ samples are  corrupted by pepper noise;
(ii) the middle chunk of samples of length  $(1-\rho )n$ consists of noise-free pixels; 
and (iii) the last chunk of samples of length  $(1-\rho_1) \rho n$ consists of pixels  corrupted by salt noise.

\noindent {\bf Low-Noise Regime, $\rho <0.5$.} 
In this regime, the noise-free pixels are the most common or typical.   
Now, recall that the entropy can be interpreted as the average rate at which a stochastic source produces information, where typical events are assigned less weight than extreme probability events.
Hence, we expect  that $\mathsf{r}_1(i,X^n)$ is smaller for values of $i$ that fall in the middle chunk of samples 
(that consists of noise-free pixels) compared to values of $i$ corresponding to other samples.
Hence, in this regime, we choose the coefficients of the L-estimator to be inversely proportional to $\mathsf{r}_1(\cdot,\cdot)$ as shown in~\eqref{eq:FilterCoefficientsSequentail} for $\rho<0.5$, where
the normalization is needed to ensure that the estimator is unbiased.

As an example, we consider a low-noise regime with $\rho=0.3$ and $\rho_1=0.05$, where we expect that roughly $30 \%$ of the image is corrupted by noise and the noise is  mostly salt.  
In Fig.~\ref{fig:ExampleOfr1AndPMF}, we plot the measure $\mathsf{r}_1(i,X^n)$ for $i \in [n]$ and $n=16$ (i.e., $4 \times 4$ window). 
Observe that in this regime,
\begin{figure}
\center 
\begin{subfigure}[c]{0.45\textwidth}
%
%
\definecolor{mycolor1}{rgb}{0.00000,0.44700,0.74100}%
\begin{tikzpicture}

\begin{axis}[%
width=6cm,
height=3.2cm,
at={(1.011in,0.642in)},
scale only axis,
xmin=0,
xmax=16,
xlabel style={font=\color{white!15!black}},
xlabel={$i$},
ymin=0,
ymax=1.01,
ylabel style={font=\color{white!15!black}},
ylabel={$\mathsf{r}_1(i,X^n)$},
axis background/.style={fill=white},
xmajorgrids,
ymajorgrids,
legend style={legend cell align=left, align=left, draw=white!15!black}
]
\addplot[ycomb, color=black, mark=o, mark options={solid, mycolor1}] table[row sep=crcr] {%
1	0.750560233063741\\
2	0.160578607478497\\
3	0.0175004364580821\\
4	0.00151449498635231\\
5	0.00224070667718164\\
6	0.0111832012702059\\
7	0.0435739415779589\\
8	0.131885620584077\\
9	0.313720195935091\\
10	0.589217806878933\\
11	0.870981880352116\\
12	0.999970726780352\\
13	0.867211892089199\\
14	0.538976926864477\\
15	0.216095992321314\\
16	0.0428435186696111\\
};
\addplot[forget plot, color=white!15!black] table[row sep=crcr] {%
0	0\\
16	0\\
};

\end{axis}

\end{tikzpicture}%
	\end{subfigure} 
	\\
	\begin{subfigure}[c]{0.45\textwidth}
%
%
\definecolor{mycolor1}{rgb}{0.00000,0.44700,0.74100}%
\begin{tikzpicture}

\begin{axis}[%
width=6cm,
height=3.2cm,
at={(1.011in,0.642in)},
scale only axis,
xmin=-5,
xmax=260,
xlabel={$\text{Support of} \ X$},
label={pmf},
ymin=0,
ymax=0.81,
ylabel={pmf},
axis background/.style={fill=white},
title style={font=\bfseries},
xmajorgrids,
ymajorgrids,
legend style={legend cell align=left, align=left, draw=white!15!black}
]
\addplot[ycomb, color=black, mark=o, mark options={solid, mycolor1}] table[row sep=crcr] {%
0	0.015\\
150	0.7\\
255	0.285\\
};
\addplot[forget plot, color=white!15!black] table[row sep=crcr] {%
0	0\\
300	0\\
};

\end{axis}

\end{tikzpicture}%
%
	\end{subfigure}%
	\vspace{-1mm}
     	\caption{$\rho=0.3; \rho_1=0.05$. \textbf{Above:} $\mathsf{r}_1(i,X^n)$ for $i \in [n], n=16$; \textbf{Below:} pmf of $X$.}
	\label{fig:ExampleOfr1AndPMF}
	\vspace{-0.55cm}
\end{figure}
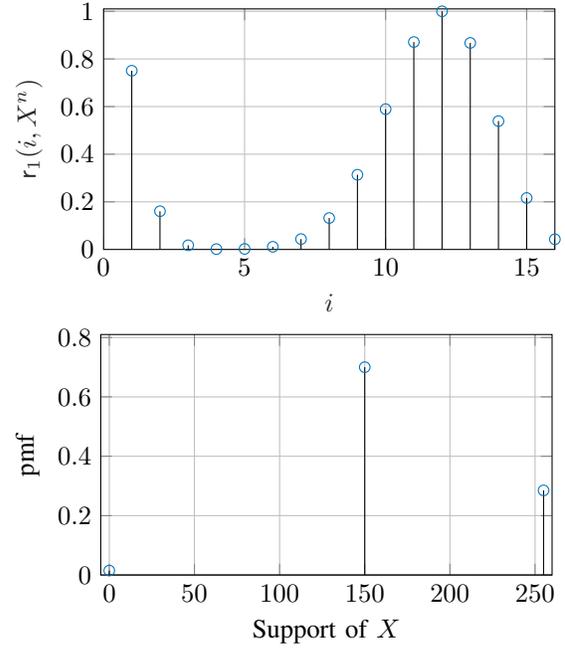
approximately $0.24$ samples are corrupted by pepper noise,
$4.56$ samples are corrupted by salt noise, and $11.2$ samples are noise-free.

  \begin{figure*}[t]
	\begin{subfigure}[t]{0.2\textwidth} \centering
	\fbox{\includegraphics[width=2cm]{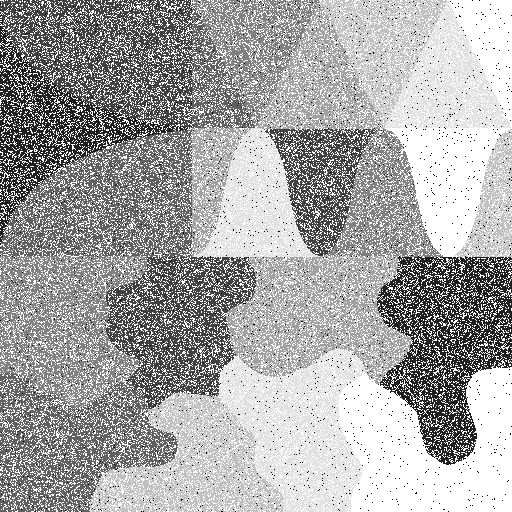}}%
	\caption{\tiny Noisy Image. \\ MSE=0.022, PSNR=10.510, \\
	SSIM=0.099, IQI=0.037.}
	\label{fig:NoisyImage}
	\end{subfigure}%
	\begin{subfigure}[t]{0.2\textwidth} \centering
	\fbox{\includegraphics[width=2cm]{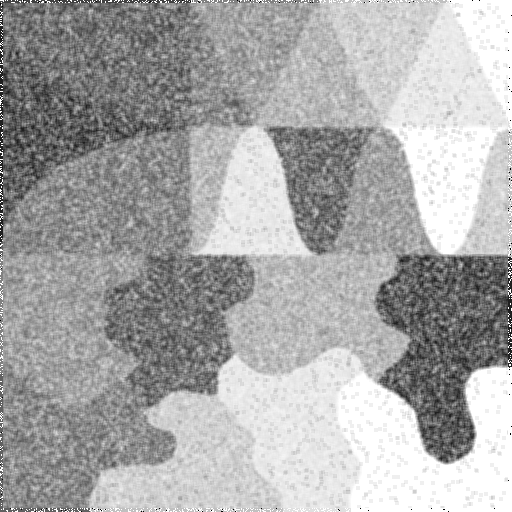}}%
	\caption{\tiny Average Filter. \\MSE=0.007, PSNR= 15.398, \\
	SSIM=0.366, IQI=0.062.}
	\label{fig:MeanFilter}
	\end{subfigure}%
	\begin{subfigure}[t]{0.2\textwidth} \centering
	\fbox{\includegraphics[width=2cm]{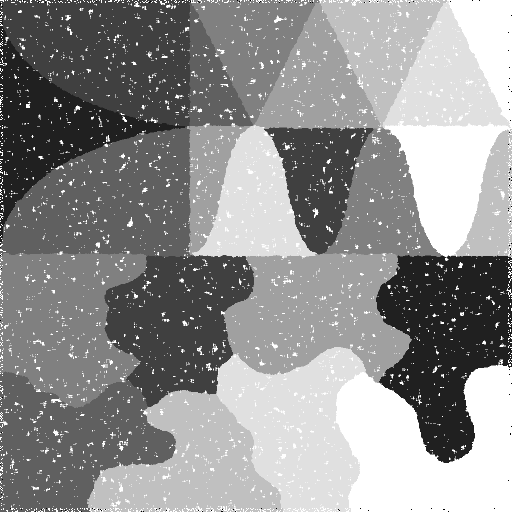}}%
	\caption{\tiny  Median Filter. \\MSE=0.003, PSNR=19.375, \\
	SSIM=0.560, IQI=0.664.}
	\label{fig:MedianFilter}
	\end{subfigure}%
	\begin{subfigure}[t]{0.2\textwidth} \centering
	\fbox{\includegraphics[width=2cm]{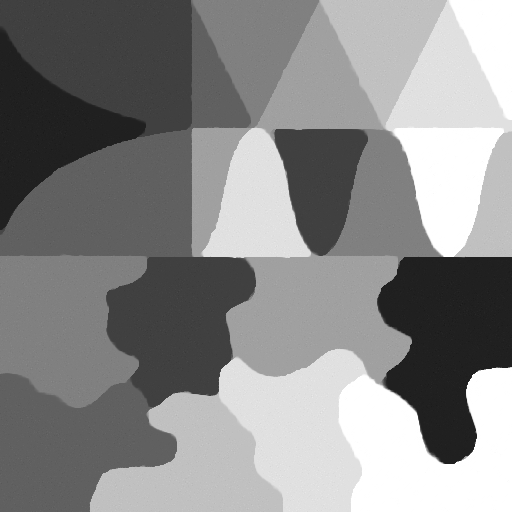}}%
	\caption{\tiny Total Variation Filter. \\$\text{MSE}=1.74\cdot10^{-4}$, PSNR=31.537, \\
	SSMI=0.956, IQI=0.130.}
	\end{subfigure}%
	\begin{subfigure}[t]{0.2\textwidth} \centering
	\fbox{\includegraphics[width=2cm]{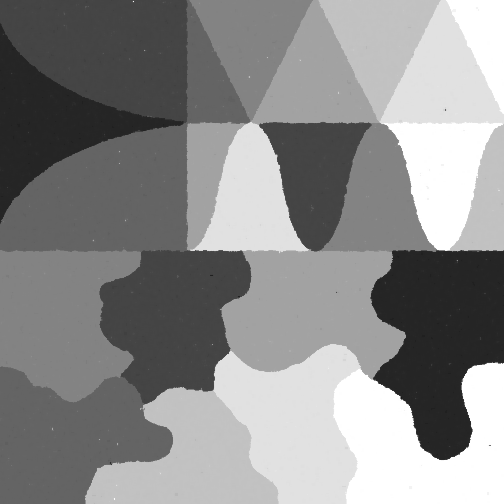}}%
	\caption{\tiny L-Estimator in~\eqref{eq:FilterCoefficientsSequentail}. \\$\text{MSE}=6.95\cdot10^{-4}$, PSNR=25.525, \\
	SSMI=0.914, IQI=0.779.}
	\end{subfigure}\\
	\text{} \hspace{3.3cm}
	\begin{subfigure}[t]{0.2\textwidth} \centering
	\fbox{\includegraphics[width=2cm]{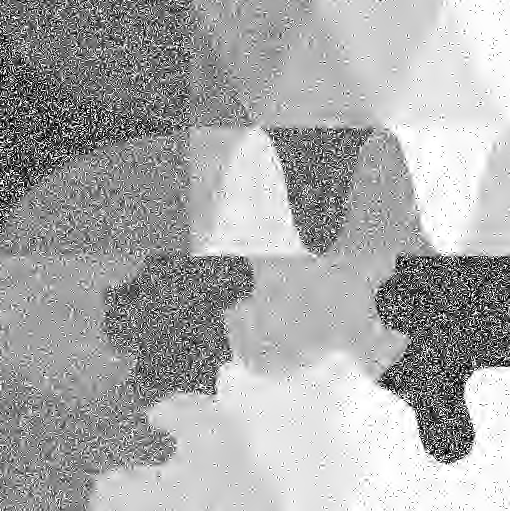}}%
	\caption{\tiny FDR filter. \\$\text{MSE}=0.012$, PSNR=13.235, \\
	SSMI=0.318, IQI=0.042.}
	\end{subfigure}%
	\begin{subfigure}[t]{0.2\textwidth} \centering
	\fbox{\includegraphics[width=2cm]{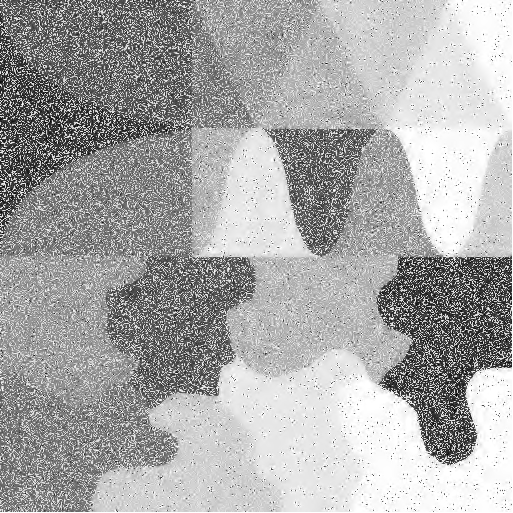}}%
	\caption{\tiny SURE filter. \\$\text{MSE}=0.013 $, PSNR=12.972, \\
	SSMI=0.195, IQI=0.045.}
	\end{subfigure}
		\begin{subfigure}[t]{0.2\textwidth} \centering
	\fbox{\includegraphics[width=2cm]{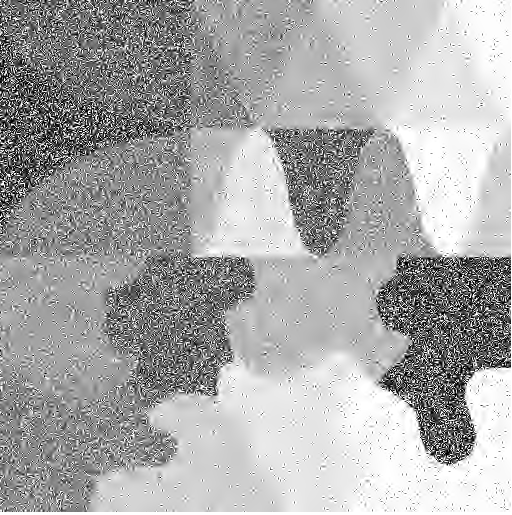}}%
	\caption{\tiny Bayes filter. \\
	$\text{MSE}=0.011$, PSNR=13.643, \\
	SSMI=0.303 , IQI=0.044.}
	\end{subfigure}\
     	\caption{Denoising salt \& pepper noise with $\rho=0.3, \rho_1=0.05$. }
	\label{fig:ExampleSmallNoise}
\end{figure*}

We now show that in the low-noise regime, our procedure in Algorithm~\ref{algo:Deno} competes with some of the state-of-the-art filters.  
The simulation results are presented in Fig.~\ref{fig:ExampleSmallNoise}.  In the simulation,
estimated values of the parameters $\rho$ and $\rho_1$ are used to train the L-estimator.  The estimates are computed as in~\eqref{eq:EstProb} and are given by $\hat{\rho}=0.3007$ and $\hat{\rho}_1=0.0508$ (recall the true values are $\rho=0.3$  and $\rho_1=0.05$). 
The coefficients of the L-estimator in~\eqref{eq:L-filter} are computed by using~\eqref{eq:FilterCoefficientsSequentail}  for $\rho < 0.5$ where the values of $\mathsf{r}_1(\cdot,\cdot)$ are those in Fig.~\ref{fig:ExampleOfr1AndPMF}.
From  Fig.~\ref{fig:ExampleSmallNoise}, we {observe that 
we have the} following performance across the four considered metrics:
\begin{align*}
&\textbf{ MSE: }  \text{Total Variation}\succ  \text{L-Estimator} \succ   \text{Median Filter} \succ \\
& \text{Avg. Filter} \succ   \text{E. Bayes Filter} \succ  \text{FDR Filter} \succ  \text{SURE Filter}, \\
&\textbf{ PSNR: }   \text{Total Variation} \succ   \text{L-Estimator} \succ   \text{Median Filter} \succ \\
& \text{Avg. Filter} \succ   \text{E. Bayes Filter} \succ  \text{FDR Filter} \succ  \text{SURE Filter}, \\
&\textbf{ SSIM: }     \text{Total Variation} \succ   \text{L-Estimator} \succ   \text{Median Filter} \succ\\
& \text{Avg. Filter} \succ   \text{FDR Filter} \succ  \text{E. Bayes Filter} \succ   \text{SURE Filter}, \\
&\textbf{ IQI: }    \text{L-Estimator}   \succ  \text{Median Filter}  \succ  \text{Total Variation} \succ \\%
& \text{Avg. Filter} \succ  \text{SURE Filter}  \succ  \text{E. Bayes Filter} \succ  \text{FDR Filter},
\end{align*} 
where, for a given metric $M$, the notation $A \succ B$ means that $A$ outperforms $B$ when $M$ is considered.
The fact that the median outperforms the total variation when the IQI metric is considered stems from the fact that the median filter allows for a better edge recovery compared to the total variation filter.
Moreover,  the L-estimator outperforms the median for all considered metrics, and has a competitive performance to that of the total variation filter (i.e., the performance is slightly worse over the MSE, PSNR and SSIM metrics, but significantly better over the IQI metric).   Finally, the L-estimator outperforms the wavelet-based filters over all metrics.

\smallskip
\noindent{\bf High-Noise Regime, $\rho \geq 0.5$.} 
Arguably, the noise-dominated regime is the most interesting case both
\begin{figure}
\center 
\begin{subfigure}[t]{0.45\textwidth}
%
%
\definecolor{mycolor1}{rgb}{0.00000,0.44700,0.74100}%
\begin{tikzpicture}

\begin{axis}[%
width=6cm,
height=3.2cm,
at={(1.011in,0.642in)},
scale only axis,
xmin=0,
xmax=37,
xlabel style={font=\color{white!15!black}},
xlabel={$i$},
ymin=0,
ymax=1.2,
ylabel style={font=\color{white!15!black}},
ylabel={$\mathsf{r}_1(i,X^n)$},
axis background/.style={fill=white},
xmajorgrids,
ymajorgrids,
legend style={legend cell align=left, align=left, draw=white!15!black}
]
\addplot[ycomb, color=black, mark=o, mark options={solid, mycolor1}] table[row sep=crcr] {%
1	0.00268370785591857\\
2	0.0214790611255973\\
3	0.0865940704296006\\
4	0.231279249746004\\
5	0.458456967378913\\
6	0.717900048823131\\
7	0.923284041273873\\
8	1.00126747140822\\
9	0.933929124220493\\
10	0.764080898806018\\
11	0.56664274044359\\
12	0.412200740576884\\
13	0.344643369238846\\
14	0.376539461562697\\
15	0.493827402545022\\
16	0.662153812639272\\
17	0.834005246506547\\
18	0.959390667180407\\
19	1.00013709987551\\
20	0.942642065290249\\
21	0.802066166808987\\
22	0.614877197719313\\
23	0.423472047842768\\
24	0.260972943434633\\
25	0.143175469634375\\
26	0.0694808585241886\\
27	0.0295934320541898\\
28	0.0109580507559418\\
29	0.00348647742654579\\
30	0.000939080330296889\\
31	0.000209993335923557\\
32	3.79532227686012e-05\\
33	5.3314161455513e-06\\
34	5.46971747106557e-07\\
35	3.6566422606057e-08\\
36	1.20490273416725e-09\\
};
\addplot[forget plot, color=white!15!black] table[row sep=crcr] {%
0	0\\
40	0\\
};

\end{axis}
\end{tikzpicture}%
%
	\end{subfigure}
~	
	\vspace{0.5cm} 
	\begin{subfigure}[t]{0.45\textwidth}
%
%
\definecolor{mycolor1}{rgb}{0.00000,0.44700,0.74100}%
\begin{tikzpicture}

\begin{axis}[%
width=6cm,
height=3.2cm,
at={(1.011in,0.642in)},
scale only axis,
xmin=0,
xmax=300,
ymin=0,
ymax=0.55,
xlabel={$\text{Support of} \ X$},
ylabel={pmf},
axis background/.style={fill=white},
xmajorgrids,
ymajorgrids,
legend style={legend cell align=left, align=left, draw=white!15!black}
]
\addplot[ycomb, color=black, mark=o, mark options={solid, mycolor1}] table[row sep=crcr] {%
0	0.21\\
150	0.3\\
255	0.49\\
};
\addplot[forget plot, color=white!15!black] table[row sep=crcr] {%
0	0\\
300	0\\
};

\end{axis}

\end{tikzpicture}%
%
	\end{subfigure}%
	\vspace{-5mm}
     	\caption{$\rho=0.7; \rho_1=0.3$. \textbf{Above:} $\mathsf{r}_1(i,X^n)$ for $i \in [n], n=36$; \textbf{Below:} pmf of $X$.}
	\label{fig:ExampleOfr1AndPMFHighNoise}
	\vspace{-0.55cm}
\end{figure}
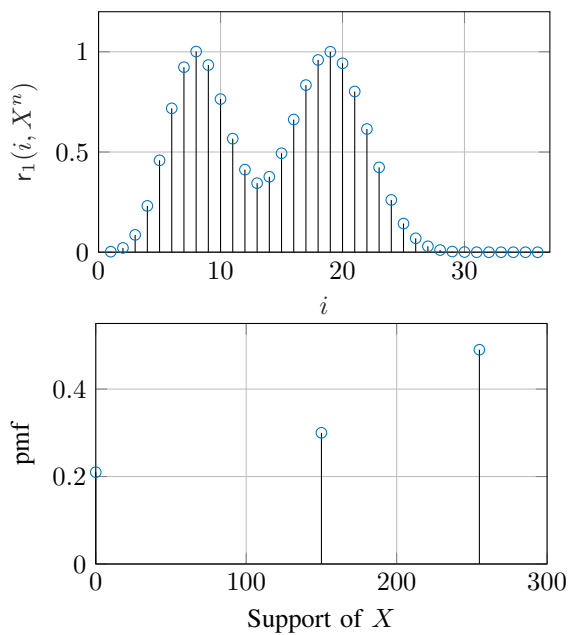
theoretically and practically. 
Consider, $\rho=0.7$ and $\rho_1=0.3$, where we expect that $70 \%$ of the image
is corrupted  by mostly salt noise. 
In Fig.~\ref{fig:ExampleOfr1AndPMFHighNoise}, we plot $\mathsf{r}_1(i,X^{n})$ for $n=36$ (i.e., $6 \times 6$ window).  
Here, approximately $7.56$ samples are corrupted by pepper noise, $17.64$ samples are corrupted by salt noise, and $10.8$ samples are noise-free. 
Thus, noisy pixels are the most common, which is a fundamental difference from the low-noise regime,
and justifies our choice of the L-estimator coefficients in~\eqref{eq:FilterCoefficientsSequentail} for $\rho \geq 0.5$. In other words, these coefficients are chosen to be directly 
proportional to $\mathsf{r}_1(\cdot,\cdot)$.
The performance of the proposed filter is evaluated in Fig.~\ref{fig:ExampleSPNoise} (top (a)-(h)), where the estimates of  $\rho$ and $\rho_1$ are computed from~\eqref{eq:EstProb} and given by $\hat{\rho}= 0.7003$ and $\hat{\rho}_1= 0.2995$.   

\begin{figure*}[t]
        \begin{subfigure}[t]{0.2\textwidth}
	\centering
	\fbox{\includegraphics[width=2cm]{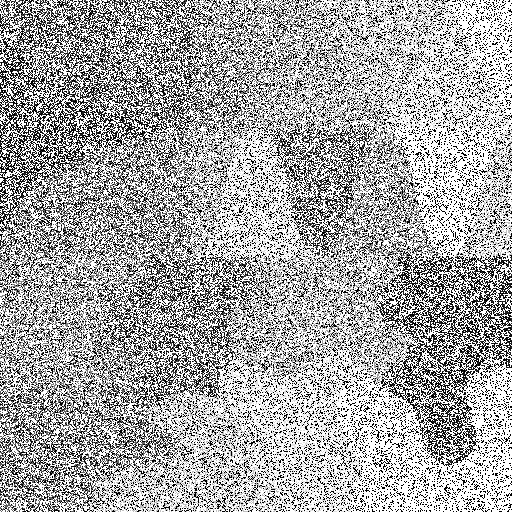}}%
	\caption{\tiny Noisy Image. \\ MSE=0.055, PSNR=6.548, \\
	SSIM=0.010, IQI=0.007.}
	\end{subfigure}%
	\begin{subfigure}[t]{0.2\textwidth}
	\centering
	\fbox{\includegraphics[width=2cm]{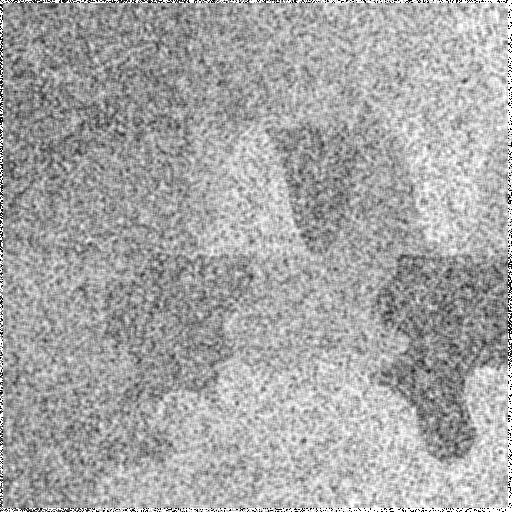}}%
	\caption{\tiny Average Filter. \\MSE=0.015, PSNR=12.177, \\
	SSIM=0.205, IQI=0.019.}
	\end{subfigure}
	\begin{subfigure}[t]{0.2\textwidth}
	\centering
	\fbox{\includegraphics[width=2cm]{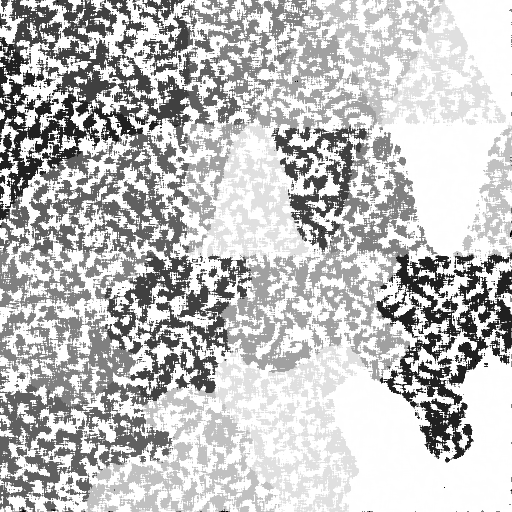}}%
	\caption{\tiny Median Filter. \\ MSE=0.032, PSNR=8.958, \\
	SSIM=0.190, IQI=0.016.}
	\end{subfigure}%
	\begin{subfigure}[t]{0.2\textwidth}
	\centering
	\fbox{\includegraphics[width=2cm]{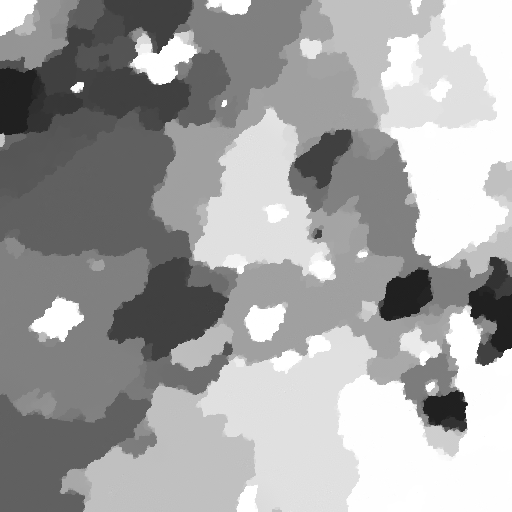}}%
	\caption{\tiny Total Variation Filter. \\ MSE= 0.021, PSNR=10.666, \\
	SSIM=0.750, IQI=0.136.}
	\end{subfigure}%
	\begin{subfigure}[t]{0.2\textwidth}
	\centering
	\fbox{\includegraphics[width=2cm]{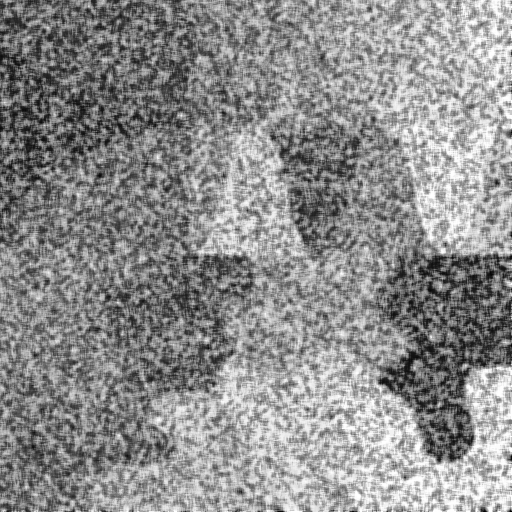}}%
	\caption{\tiny L-Estimator in~\eqref{eq:FilterCoefficientsSequentail}. \\
	MSE=0.010, PSNR=14.043,  \\ 
	SSIM= 0.114, IQI=0.021.}
	\end{subfigure} \\
	\text{}
	 \hspace{3.3cm}
	\begin{subfigure}[t]{0.2\textwidth}
	\centering
	\fbox{\includegraphics[width=2cm]{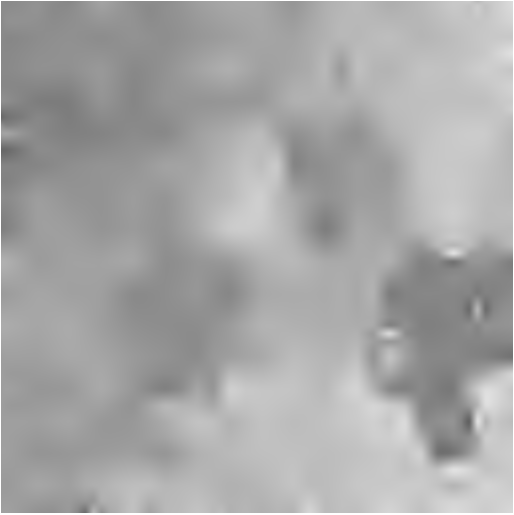}}%
	\caption{\tiny FDR filter.\\ MSE=0.014, PSNR=12.593, \\
	SSIM=0.770, IQI=0.021.}
	\end{subfigure} 
	\begin{subfigure}[t]{0.2\textwidth}
	\centering
	\fbox{\includegraphics[width=2cm]{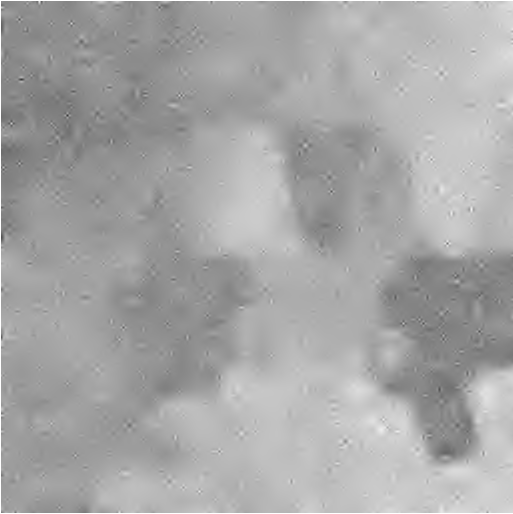}}%
	\caption{\tiny SURE filter.  \\
	MSE=0.014, PSNR=12.620,  \\
	SSIM= 0.726, IQI=0.016.}
	\end{subfigure} 
	\begin{subfigure}[t]{0.23\textwidth}
	\centering
	\fbox{\includegraphics[width=2cm]{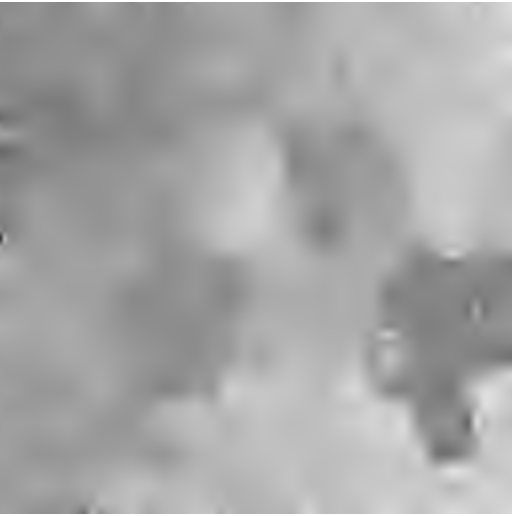}}%
	\caption{\tiny  Empirical Bayes Filter. \\ MSE= 0.014, PSNR=12.599, \\ SSIM=0.771, IQI=0.021.}
	\end{subfigure}\\
	\begin{subfigure}[t]{0.2\textwidth}
\centering
	\fbox{\includegraphics[width=2cm]{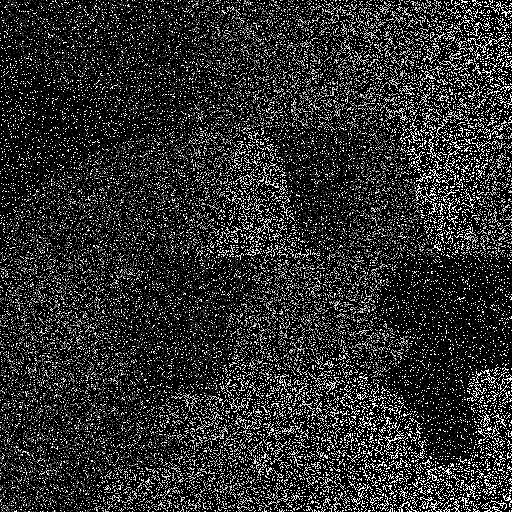}}%
	\caption{\tiny Noisy Image.\\ MSE=0.072, PSNR=5.384, \\
	SSIM=0.010, IQI=0.005.}
	\end{subfigure}%
	\begin{subfigure}[t]{0.2\textwidth}
\centering
	\fbox{\includegraphics[width=2cm]{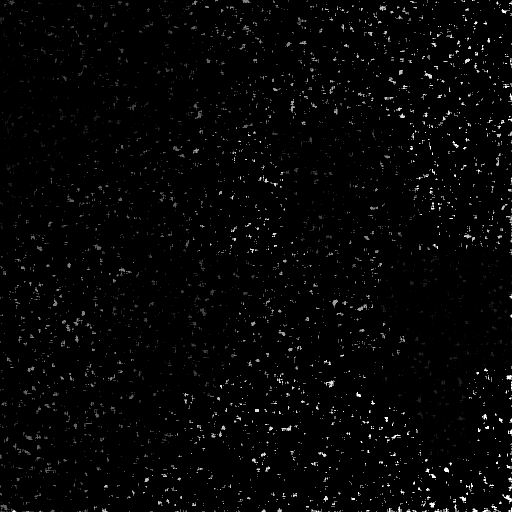}}%
	\caption{\tiny Median Filter.\\ MSE=0.088, PSNR=4.517, \\
	SSIM=0.010, IQI=0.000.}
	\end{subfigure}%
	\begin{subfigure}[t]{0.2\textwidth}
\centering
	\fbox{\includegraphics[width=2cm]{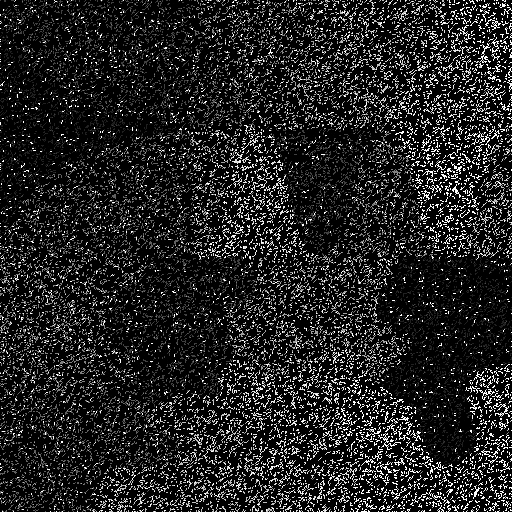}}%
	\caption{\tiny Total Variation Filter.\\ MSE=0.073, PSNR=5.336, \\
	SSIM=0.017, IQI=0.005.}
	\end{subfigure}%
	\begin{subfigure}[t]{0.2\textwidth}
	\centering
	\fbox{\includegraphics[width=2cm]{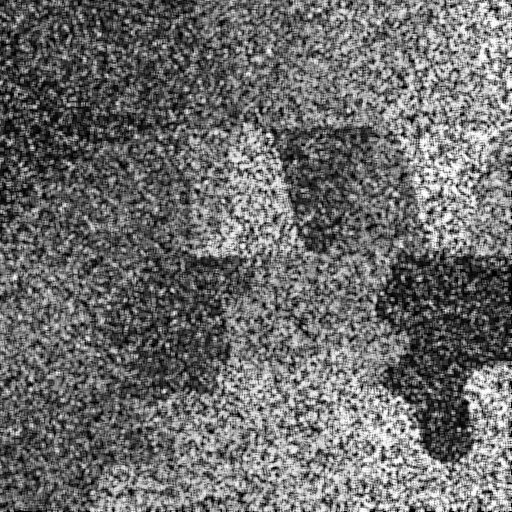}}%
	\caption{\tiny L-Estimator in~\eqref{eq:FilterCoefficientsSequentail}. \\
	MSE=0.018, PSNR=11.505, \\
	SSIM=0.061, IQI=0.019.}
	\end{subfigure}%
	\begin{subfigure}[t]{0.2\textwidth}
	\centering
	\fbox{\includegraphics[width=2cm]{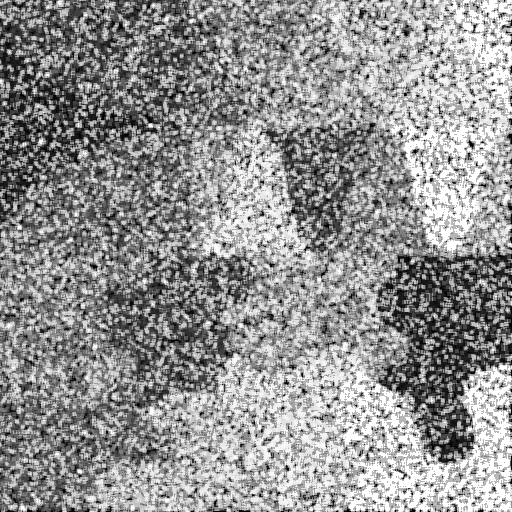}}%
	\caption{\tiny Sequential L-Estimator in~\eqref{eq:SequentialApproach}. \\
	MSE=0.016, PSNR=11.933, \\
	SSIM=0.048, IQI=0.019.}
	\end{subfigure}\\
	\text{}
	  \hspace{3.3cm}
	\begin{subfigure}[t]{0.2\textwidth}
	\centering
	\fbox{\includegraphics[width=2cm]{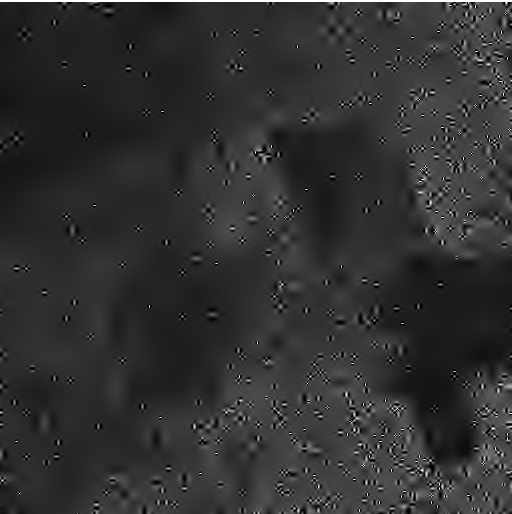}}%
	\caption{\tiny FDR filter. \\
	MSE=0.044, PSNR=7.487, \\ SSIM=0.498, IQI=0.016.}
	\end{subfigure}%
	\begin{subfigure}[t]{0.2\textwidth}
	\centering
	\fbox{\includegraphics[width=2cm]{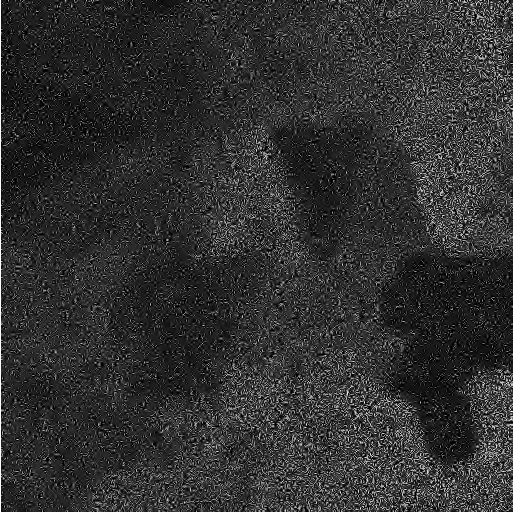}}%
	\caption{\tiny SURE filter. \\
	MSE=0.048, PSNR=7.166, \\ SSIM=0.078, IQI=0.006.}
	\end{subfigure}%
	\begin{subfigure}[t]{0.2\textwidth}
	\centering
	\fbox{\includegraphics[width=2cm]{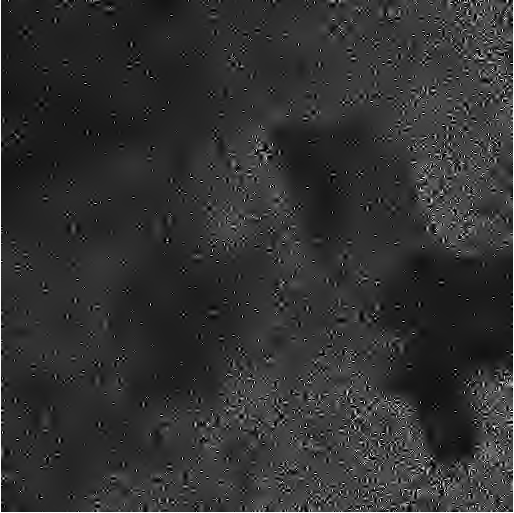}}%
	\caption{\tiny Empirical Bayes Filter. \\
	MSE=0.045, PSNR=7.445, SSIM=0.408, IQI=0.009.}
	\end{subfigure}%
     	\caption{Denoising salt \& pepper noise. 
	$\rho=0.7, \rho_1=0.3$ \textbf{top (a)-(h)}; $\rho=0.8, \rho_1=0.9$ \textbf{bottom (i)-(p)}.}
	\label{fig:ExampleSPNoise}
	\vspace{-5mm}
\end{figure*}

We observe that the median filter performs the worst for all the four considered image quality metrics, except for the SSIM metric where it outperforms the L-estimator.
This is expected since the median filter performance  degrades once the majority of the samples is corrupted. 
We have the following performance across the metrics: 
\begin{align*}
&\textbf{ MSE, PSNR: }  \text{L-Estimator} \succ   \text{FDR filter}=  \text{E. Bayes filter}  \\
&\quad \quad \quad = \text{SURE filter}  \succ  \text{Avg.\ Filter} \succ \text{Total Variation}   ,\\
&\textbf{ SSIM: } \text{E. Bayes filter}  \succ  \text{FDR filter} \succ  \text{Total Variation}  \\
& \quad \quad \quad \succ  \text{SURE filter}  \succ   \text{Avg.\ Filter} \succ \text{Median Filter},\\
&\textbf{ IQI: }    \text{Total Variation}     \succ  \text{FDR filter}= \text{E. Bayes filter} \\
&  \quad \quad \quad    \succ  \text{L-Estimator}   \succ   \text{Avg.\ Filter} \succ  \text{SURE filter},
\end{align*}  
where, for a given metric $M$, the notation $A \succ B$ means that $A$ outperforms $B$ when $M$ is considered.
The result above suggests that the L-estimator is very much competitive with the total variation filter and wavelet-based filters, and most of the time it also outperforms the average filter. 
It is also worth noting that visually the total variation filter appears to have the worst performance across all the four filters in terms of recovering the shapes, but this observation is not captured by the SSIM and IQI metrics.

\begin{figure*}[t]	
	\begin{subfigure}[t]{0.2\textwidth}
	\centering
	\fbox{\includegraphics[width=2cm]{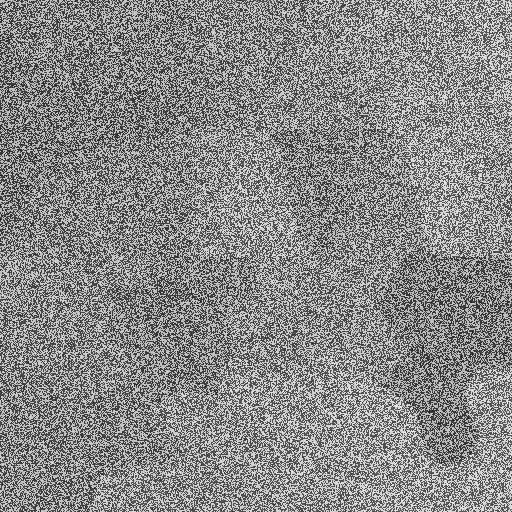}}%
	\caption{\tiny Noisy Image.\\ MSE=0.032, PSNR=8.917, \\
	SSIM=0.020, IQI=0.005.}         
	\end{subfigure}%
	\begin{subfigure}[t]{0.2\textwidth}\centering
	\fbox{\includegraphics[width=2cm]{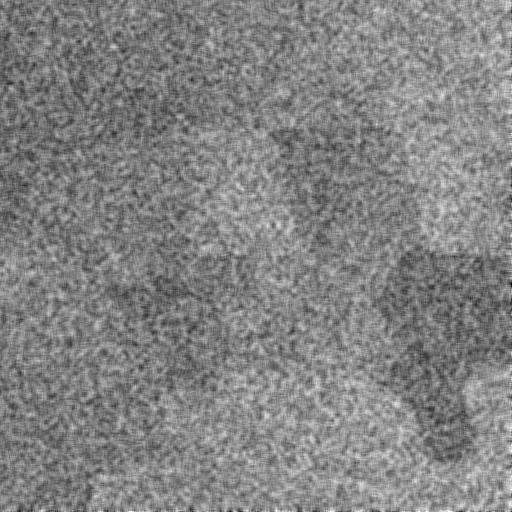}}%
	\caption{ \tiny Mean Filter.\\ MSE=0.015, PSNR=12.141, \\
	SSIM=0.221, IQI=0.012.}          
	\end{subfigure}%
	\begin{subfigure}[t]{0.2\textwidth}\centering
	\fbox{\includegraphics[width=2cm]{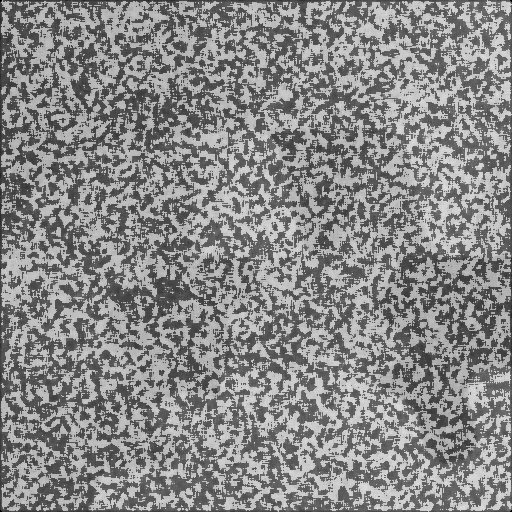}}%
	\caption{\tiny Median Filter.\\ MSE=0.026, PSNR=9.794, \\
	SSIM=0.056, IQI=0.002.
	}
	\end{subfigure}%
	\begin{subfigure}[t]{0.2\textwidth}\centering
	\fbox{\includegraphics[width=2cm]{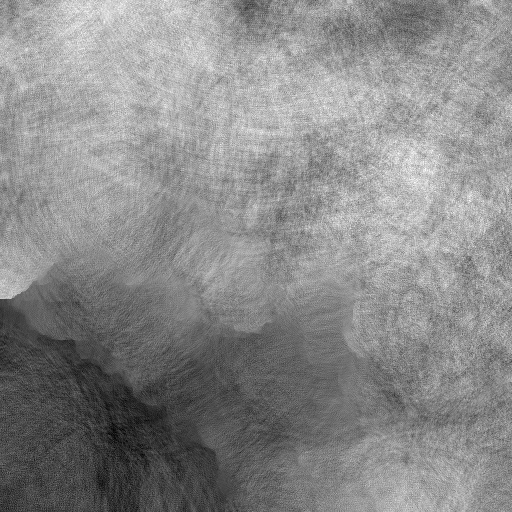}}%
	\caption{ \tiny Total Variation Filter .\\
	MSE=0.029, PSNR=9.382, \\
	SSIM=0.237, IQI=-0.001.
	}
	\end{subfigure}%
	\begin{subfigure}[t]{0.2\textwidth} \centering
	\fbox{\includegraphics[width=2cm]{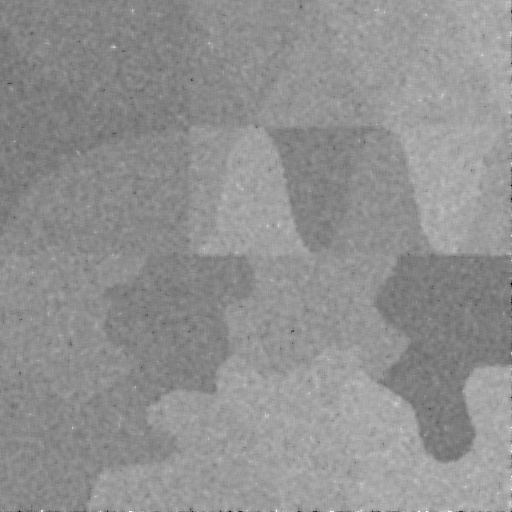}}%
	\caption{ \tiny  L-Estimator. \\
	MSE=0.014, PSNR=12.533, \\
	SSIM=0.627, IQI=0.021.}
	\end{subfigure} \\
	\text{}
	  \hspace{3.3cm}
	\begin{subfigure}[t]{0.2\textwidth} \centering
	\fbox{\includegraphics[width=2cm]{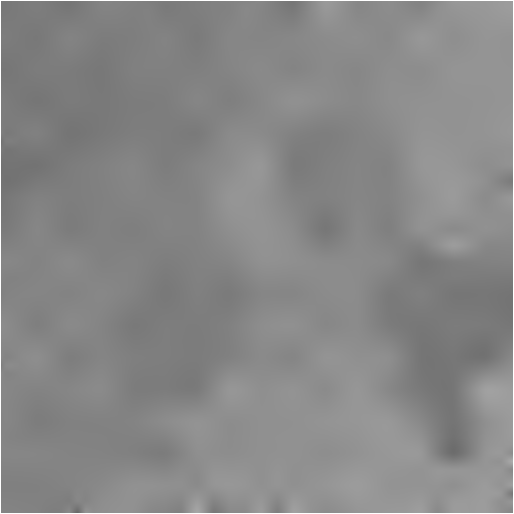}}%
	\caption{ \tiny  FDR filter \\
	MSE=0.016, PSNR=11.960, \\
	SSIM=0.780, IQI=0.048.}
	\end{subfigure} 	
	\begin{subfigure}[t]{0.2\textwidth} \centering
	\fbox{\includegraphics[width=2cm]{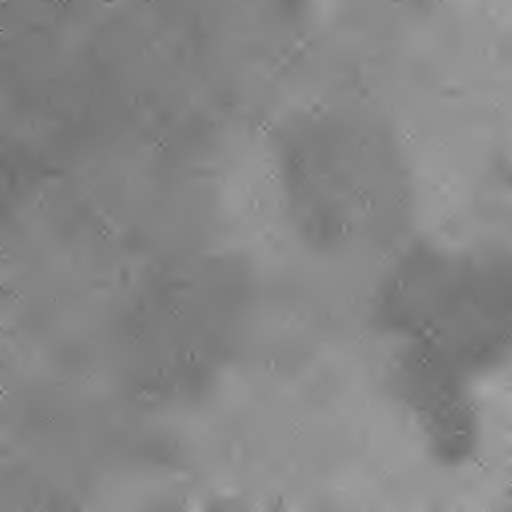}}%
	\caption{ \tiny  SURE filter \\
	MSE=0.016, PSNR=11.977, \\
	SSIM=0.775, IQI=0.016.}
	\end{subfigure} 	
	\begin{subfigure}[t]{0.2\textwidth} \centering
	\fbox{\includegraphics[width=2cm]{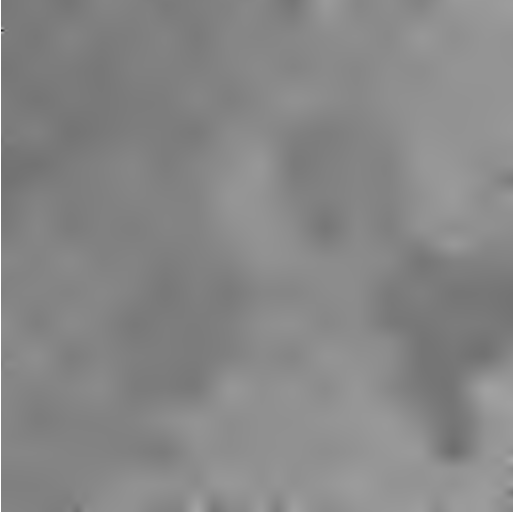}}%
	\caption{ \tiny Empirical Bayes filter \\
	MSE=0.016, PSNR=11.960, \\
	SSIM=0.781, IQI=0.048.}
	\end{subfigure} 		
     	\caption{Denoising mixed Gaussian noise with  $\mu=[ -2, 2]$, $\sigma^2=[0.15, 0.1]$  and $p=[0.5, 0.5]$.  }
	\label{fig:MixedExample}
	\vspace{-5mm}
\end{figure*}

In the extremely high-noise regime, we observe through extensive simulations that the  L-estimator performs significantly better than the total variation filter and the wavelet-based filters.
 The two bottom rows of Fig.~\ref{fig:ExampleSPNoise}~(i)-(p) show the filters performance for $n=16$ in an extremely noisy setting where $\rho=0.8$ and $\rho_1=0.9$. Here, we expect that $80\%$ of the image is corrupted by noise, and this perturbation is dominated by pepper noise.     
The estimates of  $\rho$ and $\rho_1$ are computed from~\eqref{eq:EstProb} and given by $\hat{\rho}=  0.7990$ and $\hat{\rho}_1=0.8992$.
In addition to the already used filters, 
in this regime
Fig.~\ref{fig:ExampleSPNoise} also shows the L-estimator performance where the coefficients are chosen based on the sequential approach, discussed in~\eqref{eq:Cond}, i.e.,   for $  i_k \in \bar{\mathcal{S}}_{1}^S$,
\begin{align}
\begin{split}
 \alpha_k&=  
  \frac{ \mathsf{r}_1(i_k,X^n| \mathcal{V}_{k-1})}{  \sum_{k=1}^{d}  \mathsf{r}_1(i_k,X^n| \mathcal{V}_{k-1}) },  \text{ when }  k \le d,  \text{ and } \\
  \alpha_k&=    0,   \text{ when }  k>d,   \label{eq:SequentialApproach}
  \end{split}
\end{align} 
where $\mathcal{V}_{k-1}$ is the set that contains the first $k-1$ indices of $ \bar{\mathcal{S}}_{1}^S$ and $d$ is the truncation parameter.   
The idea is to choose the $k$-th coefficient by conditioning on the information that has been already incorporated into the previously selected $k-1$ coefficients.   
We highlight that $d$ is introduced for computational purposes to speed the simulations, and with reference to Fig.~\ref{fig:ExampleSPNoise} we have $d=4$. Fig.~\ref{fig:ExampleSPNoise} shows that the total variation and median filters perform on the level of the noisy image.   
We also note that the L-estimator with coefficients as in~\eqref{eq:SequentialApproach} offers better MSE and PSNR metrics than the L-estimator with coefficients as in~\eqref{eq:FilterCoefficientsSequentail}, but  
performs either the same or worse for IQI and SSIM metrics. Finally, we highlight that the L-estimator based on the joint approach in~\eqref{eq:Joint} was also simulated, and observed to offer a similar performance to the sequential L-estimator in~\eqref{eq:FilterCoefficientsSequentail}.   The performance of all filters is as follows:
\begin{align*}
&\textbf{ MSE: }  \text{Seq. L-Estimator} \succ   \text{L-Estimator} \succ   \text{FDR filter}  \\
&\quad \quad  \succ  \text{E. Bayes filter} \succ \text{SURE filter}  \succ \text{Total Variation},   \\
&\textbf{PSNR: }  \text{Seq. L-Estimator} \succ   \text{L-Estimator} \succ   \text{FDR filter}  \\
& \quad  \quad \succ  \text{E. Bayes filter}  \succ \text{SURE filter}   \succ \text{Total Variation},   \\
&\textbf{ SSIM: }  \text{FDR filter}  \succ \text{E. Bayes filter}  \succ   \text{SURE filter}    \\
& \quad  \quad   \succ    \text{L-Estimator} \succ   \text{Seq. L-Estimator} \succ   \text{Total Variation}  ,\\
&\textbf{ IQI: }   \text{Seq. L-Estimator} \succ   \text{L-Estimator} \succ   \text{FDR filter}  \\
&  \quad \quad   \succ  \text{E. Bayes filter} \succ \text{SURE filter}  \succ \text{Total Variation}     .
\end{align*}
The above suggests that the L-estimator is very much competitive with the total variation filter and wavelet-based filters.  In particular, L-estimators perform better  than wavelet-based denoisers over the MSE, PSNR and IQI metrics, and better than the total variation denoiser over all metrics. 

\subsection{Image Denoising in Additive Continuous Noise}  \label{sec:contnoise}
Now we consider image denoising under the signal model $X=x+Z,$ where $x$ is the unknown pixel
value and $Z$ is random noise.   We consider two example noise
distributions,  Cauchy and mixed Gaussian. 
In particular, here we focus on the mixed Gaussian case, and an in the next subsection we will focus on  the case when $Z$ is  Cauchy. We also performed simulations for Gaussian $Z$ and observed that the total variation filter always outperforms our  proposed L-estimator. We believe
this is due to the fact that the total variation filter was designed for Gaussian noise perturbation.

With mixed Gaussian noise, our  denoising  works as in Algorithm~\ref{algo:Deno}, but the coefficients of the L-estimator  are now chosen with respect to the $\mathsf{r}_3 (\cdot,\cdot)$ measure\footnote{Simulations were performed also for the $\mathsf{r}_2 (\cdot,\cdot)$ measure and observed to have similar performance as for the $\mathsf{r}_3 (\cdot,\cdot)$ measure.}: 
\begin{equation}
\alpha_k=\frac{ \mathsf{r}_3(k,X^n)}{  \sum_{i=1}^{n} \mathsf{r}_3(i,X^n)  },  \label{eq:CoefficientForCauchy}
\end{equation}
where $n=25$ (i.e., $5\times5$ window). The Gaussian mixture has two components with means  $\mu=[-2, 2]$, variances $\sigma^2=[0.15, 0.1]$  and weights $p=[0.5, 0.5]$.
Fig.~\ref{fig:ExampleMixed} shows $\mathsf{r}_3(i,X^n)$.
\begin{figure}[h!]  	
 \centering
\begin{subfigure}[c]{0.45\textwidth}
%
%
\definecolor{mycolor1}{rgb}{0.00000,0.44700,0.74100}%
\definecolor{mycolor2}{rgb}{0.85000,0.32500,0.09800}%
\begin{tikzpicture}

\begin{axis}[%
width=6cm,
height=3.2cm,
at={(1.011in,0.642in)},
scale only axis,
xmin=0,
xmax=25,
xlabel={$i$},
ymode=log,
log origin=infty,
ymin=0.00769930520532114,
ymax=10,
ylabel style={font=\color{white!15!black},at={(-0.2,0.5)}},
ylabel={$\mathsf{r}_3(i,X^n)$},
yminorticks=true,
axis background/.style={fill=white},
axis x line*=bottom,
axis y line*=left,
xmajorgrids,
ymajorgrids,
yminorgrids,
legend style={legend cell align=left, align=left, draw=white!15!black},
yticklabel style={/pgf/number format/fixed},
]

\addplot[ycomb, color=mycolor2, mark=o, mark options={solid, mycolor2}] table[row sep=crcr] {%
1	0.0508728612016438\\
2	0.0321975560700346\\
3	0.028323791383747\\
4	0.0301635308641056\\
5	0.0481579498835589\\
6	0.113499868682981\\
7	0.333433514850488\\
8	0.918374438128571\\
9	2.08181540947654\\
10	4.03044780403517\\
11	6.53855755152209\\
12	8.7072213979127\\
13	9.57361387673105\\
14	8.67527308889859\\
15	6.4336077860097\\
16	3.95810124957111\\
17	2.03970486351106\\
18	0.856253240000622\\
19	0.297432200732567\\
20	0.10664982312938\\
21	0.0345876609701561\\
22	0.0198857058679696\\
23	0.0190383974345928\\
24	0.0212705852040828\\
25	0.0333604298070642\\
};

\end{axis}

\end{tikzpicture}%
%
	\end{subfigure}
	~
	
	\vspace{0.5cm}
	\begin{subfigure}[c]{0.45\textwidth}
	\input{mixGvsG.tex}%
	\end{subfigure}%
     	\caption{$\mathsf{r}_3(i,X^n),$ for $i \in [n]$ and $n=25$ (\textbf{top}), and pdf of the mixed Gaussian random variable $Z$ with $\mu=\protect \begin{bmatrix} -2 & 2 \protect \end{bmatrix}$ , $\sigma^2=\protect \begin{bmatrix} 0.15 & 0.1 \protect \end{bmatrix}$  and $p=\protect \begin{bmatrix} 0.5 & 0.5 \protect \end{bmatrix}$ (\textbf{bottom}).}
\vspace{-0.55cm}
	\label{fig:ExampleMixed}
\end{figure}
Fig.~\ref{fig:MixedExample}  shows all  filters performance for this setting, assuming  known $\mu$, $\sigma$ and $p$.  The performance of all filters is as follows:
 \begin{align*}
&\textbf{ MSE: }    \text{L-Estimator} \succ \text{Avg.\ Filter} \succ   \text{SURE filter}    \\
& \quad \quad \succ  \text{E. Bayes filter} = \text{FDR filter}   \succ \text{Median Filter},   \\
&\textbf{PSNR: }   \text{L-Estimator} \succ  \text{Avg.\ Filter} \succ  \text{SURE filter}     \\
&  \quad \quad \succ  \text{FDR filter}   \succ \text{E. Bayes filter} \succ  \text{Median Filter},  \\
&\textbf{ SSIM: } \text{E. Bayes filter}  \succ  \text{FDR filter}  \succ  \text{SURE filter}   \\
&  \quad \quad  \succ \text{L-Estimator} \succ   \text{Avg.\ Filter} \succ \text{Median Filter},\\
&\textbf{ IQI: }   \text{E. Bayes filter} \succ    \text{FDR filter} \succ   \text{L-Estimator}  \\
&  \quad \quad \succ \text{SURE filter}  \succ   \text{Avg.\ Filter} \succ \text{Median Filter}  .
\end{align*}
Here the L-estimator always outperforms the total variation filter. Moreover, it outperforms all filters over the MSE and PSNR metrics and its performance is comparable to those of wavelet-based filters over the SSIM and IQI metrics.  

\begin{figure*}[t]	
	\begin{subfigure}[t]{0.2\textwidth} \centering
	\fbox{\includegraphics[width=2cm]{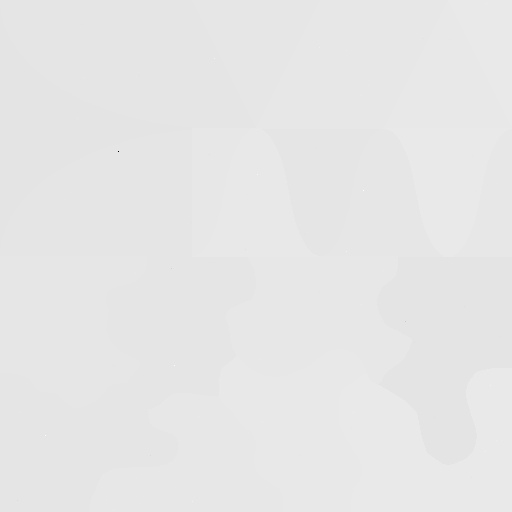}}%
	\caption{\tiny Noisy Image. \\ MSE=0.052, PSNR=6.765, \\
	SSIM=0.690, IQI=0.513}
	\end{subfigure}%
	\begin{subfigure}[t]{0.2\textwidth}\centering
	\fbox{\includegraphics[width=2cm]{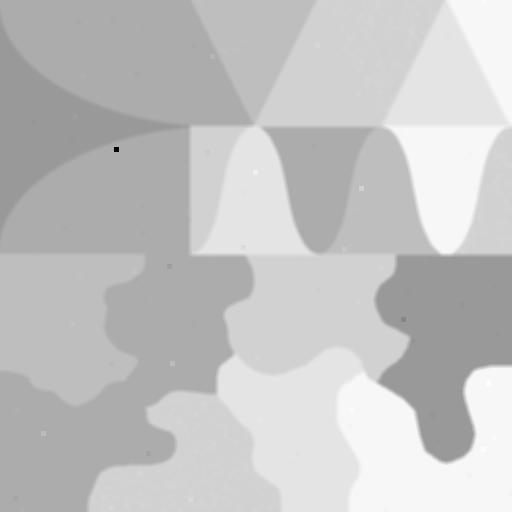}}%
	\caption{\tiny Mean Filter. \\ MSE=0.021, PSNR=10.810, \\
	SSIM=0.765, IQI=0.529.}
	\end{subfigure}%
	\begin{subfigure}[t]{0.2\textwidth}\centering
	\fbox{\includegraphics[width=2cm]{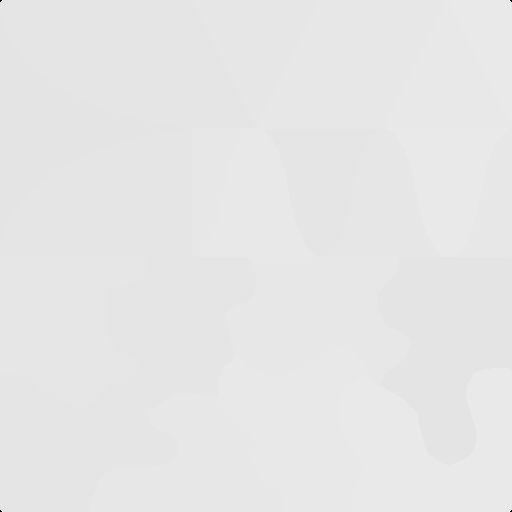}}%
	\caption{\tiny Median Filter.\\ MSE=0.052, PSNR=6.765, \\
	SSIM=0.690, IQI=0.653.
	}
	\end{subfigure}%
	\begin{subfigure}[t]{0.2\textwidth}\centering
	\fbox{\includegraphics[width=2cm]{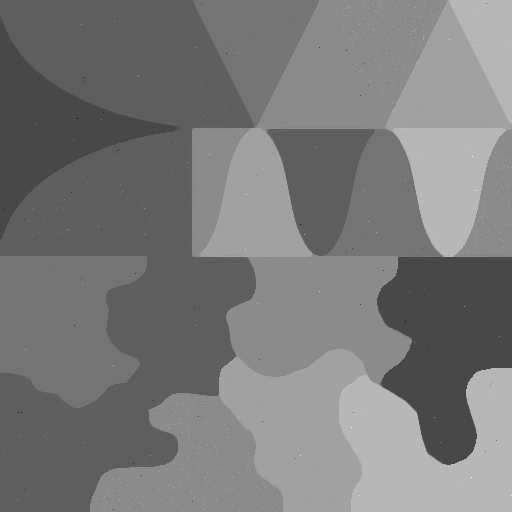}}%
	\caption{\tiny Total Variation Filter. \\
	MSE=0.007, PSNR=15.682, \\
	SSIM=0.887, IQI=0.710.
	}
	\end{subfigure}%
	\begin{subfigure}[t]{0.2\textwidth}\centering
	\fbox{\includegraphics[width=2cm]{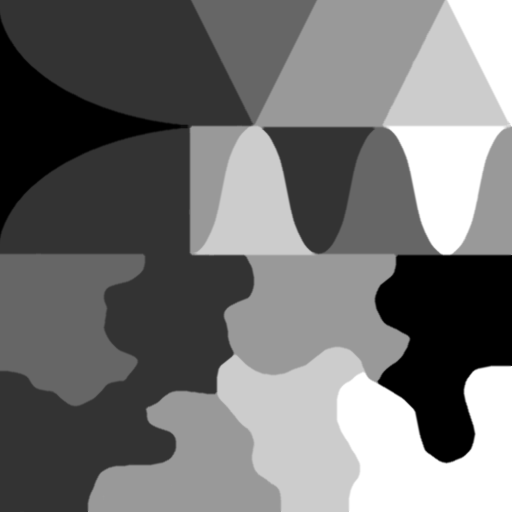}}%
	\caption{\tiny L-Estimator.\\
	MSE=0.004, PSNR=18.471, \\
	SSIM=0.770, IQI=0.732.}
	\end{subfigure}%
     	\caption{Denoising Cauchy noise with parameter $\gamma=0.0002$.}
	\label{fig:ExampleCauchyFilter}
\end{figure*}
\subsection{Cauchy Noise Distribution}
\label{app:Cauchy}
Now we consider a continuous noise model as discussed in Section~\ref{sec:contnoise}.
We let the noise, $Z$, be distributed according to a Cauchy distribution. This is a heavy tail distribution that models impulsive noise, which occurs commonly in 
image processing applications~\cite{barnett1966order}.   
In the presence of Cauchy noise, our  denoising algorithm works as in Algorithm~\ref{algo:Deno}, however, the coefficients of the L-estimator in~\eqref{eq:L-filter}  are now chosen with respect to the $\mathsf{r}_3 (\cdot,\cdot)$ measure as in~\eqref{eq:CoefficientForCauchy}.
\begin{figure}
\centering
\begin{subfigure}[t]{0.45\textwidth}
%
%
\definecolor{mycolor1}{rgb}{0.00000,0.44700,0.74100}%
\definecolor{mycolor2}{rgb}{0.85000,0.32500,0.09800}%
\begin{tikzpicture}

\begin{axis}[%
width=6cm,
height=3.2cm,
at={(1.011in,0.642in)},
scale only axis,
xmin=0,
xmax=25,
xlabel={$i$},
ymode=log,
ymin=1e-09,
ymax=0.0001,
ylabel style={font=\color{white!15!black}},
ylabel={$\mathsf{r}_3(i,X^n)$},
yminorticks=true,
log origin=infty,
axis background/.style={fill=white},
axis x line*=bottom,
axis y line*=left,
xmajorgrids,
ymajorgrids,
yminorgrids,
legend style={legend cell align=left, align=left, draw=white!15!black}
]

\addplot[ycomb, color=mycolor2, mark=o, mark options={solid, mycolor2}] table[row sep=crcr] {%
1	inf\\
2	1.16857329666749e-05\\
3	5.84488240637876e-07\\
4	1.3135300861371e-07\\
5	4.99325205256051e-08\\
6	2.52572709859822e-08\\
7	1.49805683763368e-08\\
8	9.97692681531232e-09\\
9	7.41172307298716e-09\\
10	5.91658896989971e-09\\
11	5.04783592252878e-09\\
12	4.60688978355454e-09\\
13	4.48310334220854e-09\\
14	4.61648674031603e-09\\
15	5.05567746057442e-09\\
16	5.9204874689034e-09\\
17	7.36812628400461e-09\\
18	1.00769367550745e-08\\
19	1.50223425716596e-08\\
20	2.53721263891311e-08\\
21	5.02876398592827e-08\\
22	1.33261931205095e-07\\
23	5.83547392517281e-07\\
24	1.48138129359672e-05\\
25	inf\\
};

\end{axis}

\end{tikzpicture}%
%
	\end{subfigure}
	~
	\begin{subfigure}[t]{0.45\textwidth}
	\input{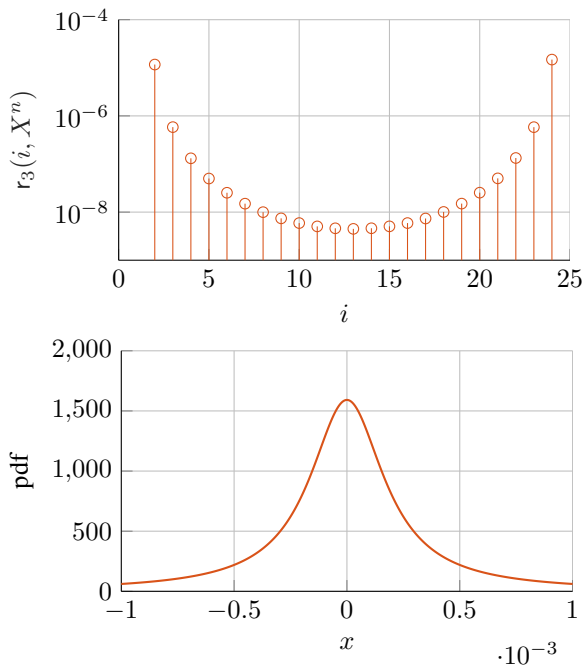}%
	\end{subfigure}%
 	\caption{Cauchy random variable $Z$ with $x_0 =0$ and $\gamma=0.0002$. \\
	\textbf{Top:} $\mathsf{r}_3(i,X^n),$ for $i \in [n]$ with $n=25$; \textbf{Bottom:} pdf of $Z$.}
         \label{fig:ExampleCauchy}
\vspace{-0.55cm}
\end{figure}

Using location parameter, $x_0 = 0$, and  scale parameter, $\gamma=0.0002$, in Fig.~\ref{fig:ExampleCauchy} we plot $\mathsf{r}_3(i,X^n)$ for $i \in \{2,\ldots, n-1\}$ and $n=25$ (i.e., $5 \times 5$ window) and the pdf of $Z$.
We highlight that $\mathsf{r}_3(1,X^n)=\mathsf{r}_3(n,X^n)=\infty$, which is due to the infinite variance of the Cauchy distribution. However, $\mathsf{r}_3(i ,X^n)< \infty$ for $i \in \{2,\ldots, n-1\}$, as we observe from Fig.~\ref{fig:ExampleCauchy}.

Fig.~\ref{fig:ExampleCauchyFilter} shows the performance of all the four filters for the case where the Cauchy scale parameter is  given by $\gamma=0.0002$, and it is assumed to be known. In this example, the L-estimator has the best  performance as compared to all other filters across all four metrics, except the SSIM where the total variation filter has a slightly better performance. It is also important to note that the MSE and PSNR  metrics might not be meaningful in this case since the Cauchy noise has infinite variance.  

\section{Conclusion}
This work has proposed an information-theoretic framework for finding the order statistic that contains the most information about the random sample. Specifically, the 
work has proposed three different information-theoretic measures to quantify the informativeness of order statistics. As an example, all three measures have been evaluated for  discrete Bernoulli and continuous Uniform random samples.
As an application, the proposed measures have been used to choose the coefficients of the L-estimator filter to denoise an image corrupted by random noise. To show the utility of our approach, several examples of various noise mechanisms (e.g., salt and pepper, mixed Gaussian) have been considered, and the proposed filters have been shown to be competitive with off-the-shelf filters (e.g., median, total variation and wavelet).

\appendices

\section{Joint Distribution of $k$ Ordered Statistics}

\subsection{Discrete Random Variables}
\label{app:Sec:ProofOfJointPMFD}
\begin{lem} 
\label{lem:OS_dist}
Let $X_1, X_2, \ldots, X_n$ be i.i.d.\ r.v. from a discrete distribution with cumulative distribution function $F(x)$.
Let $ \mathcal{S} =\{ ( i_1,i_2,\ldots,  i_k ): 1 \le  i_1< i_2< \ldots <  i_k \le n  \}$ and let 
$
P(X_{(\mathcal{S})} = x_{(\mathcal{S})}): = P(\cap_{i \in \mathcal{S}} (X_{(i)} = x_{(i)})),
$
where $x_{(i)}$ denotes the observation associated to index $i$.
Then, $P(X_{(\mathcal{S})} = x_{(\mathcal{S})})$ is non-zero only if
$
x_{(i_1)} \leq x_{(i_2)} \leq \ldots \leq x_{(i_k)}
$
and when this is true we have that
\begin{subequations}
\label{eq:JointProbNew}
\begin{align}
 &P \left (X_{(\mathcal{S})} = x_{(\mathcal{S})}  \right ) \notag\\
&=  F_{(\mathcal{S})} ( x_{(\mathcal{S})})  - \sum_{v=1}^k \Big ( (-1)^{v-1} \sum_{\substack{\mathcal{I} \subseteq \mathcal{S} \\ |\mathcal{I}| =v} } F_{(\mathcal{I}, \mathcal{I}^c)} ( x_{(\mathcal{I})}^-, x_{(\mathcal{I}^c)}  ) \Big ), \label{eq:JointProbNewa} \\
&F_{(\mathcal{I}, \mathcal{I}^c)} \left( x_{(\mathcal{I})}^-, x_{(\mathcal{I}^c)} \right ) \notag \\
& = \!\!\!\!\sum_{t_{[k]} \in \mathcal{T}}\! \prod_{j=1}^k {{n - \sum_{u=j+1}^k t_u} \choose {t_j}} g (t_{[k]}, \{x^-_{(\mathcal{I})} \cup x_{(\mathcal{I}^c)}\}  ),
\label{eq:JointProbNewb} \\
&g(t_{[k]},y_{(\mathcal{S})})\!=\!  [ 1\!-\!F(y_{(i_k)})]^{t_k}  [ F(y_{(i_k)})\!-\!F(y_{(i_{k-1})}) ]^{t_{k-1}} \ldots \notag\\
&\hspace{1cm}  [ F(y_{(i_2)})-F(y_{(i_{1})}) ]^{t_{1}}  [ F(y_{(1)}) ]^{n-\sum_{u=1}^k t_u}, \label{eq:JointProbNewc}
\end{align}
\end{subequations}
where
$
\mathcal{T} =  \{ t_{[k]} \geq 0: \sum_{m=j}^k t_m \leq n-i_j, \forall j \in [k] \},
$
with $t_{[k]} = \{t_1,t_2,\ldots, t_k\}, t_u \geq 0 \ \forall u \in[k]$.
\end{lem}

\begin{proof}
For all $t \in [k]$, we define the event
\begin{align}
\label{eq:At}
(A_t)^c = \{ X_{(i_t)} = x_{(i_t)} \, | \, X_{(i_t)} \leq x_{(i_t)}  \},
\end{align}
where $(\cdot)^c$ denotes the complement of the event.  First notice that by De Morgan's Law we have that
\begin{align}
&P \left(X_{(\mathcal{S})} = x_{(\mathcal{S})} \, |  \, X_{(\mathcal{S})} \leq x_{(\mathcal{S})} \right ) \notag \\
 &= P \left (\cap_{t =1}^k (A_t)^c \,| \,X_{(\mathcal{S})} \leq x_{(\mathcal{S})} \right) \notag \\
&= P \left ( \left (\cup_{t =1}^k A_t \right )^c | X_{(\mathcal{S})} \leq x_{(\mathcal{S})}\right) \notag \\
&=1 - P \left ( \cup_{t =1}^k A_t \, | \, X_{(\mathcal{S})} \leq x_{(\mathcal{S})}  \right).
\label{eq:DeMorgan}
\end{align}
Next we study the probability on the right side of \eqref{eq:DeMorgan}.  First, applying the inclusion-exclusion principle and, for any subset $\mathcal{I} \subseteq \mathcal{S}$, defining the event $A_{\mathcal{I}} : = \cap_{i \in \mathcal{I}} A_i$, we find
\begin{align}
&P \left ( \cup_{t =1}^k A_t \, | \, X_{(\mathcal{S})} \leq x_{(\mathcal{S})}  \right)  \notag \\
&= \sum_{t=1}^k \Big ( (-1)^{t-1} \sum_{\substack{\mathcal{I} \subseteq \mathcal{S} \\ |\mathcal{I}| =t} } P \left (A_{\mathcal{I}} \, |\,  X_{(\mathcal{S})} \leq x_{(\mathcal{S})} \right ) \Big ).
\label{eq:IncExcl}
\end{align}
Next notice that $P(\mathcal{X} \, | \, \mathcal{Y}) = P(\mathcal{X},\mathcal{Z} \, | \, \mathcal{Y})$ for $\mathcal{Z} \subseteq \mathcal{Y}$.  Then, for any set $\mathcal{I} \subseteq \mathcal{S}$, denoting $\mathcal{I}^c = \mathcal{S} \setminus \mathcal{I}$,
\begin{align}
&P (A_{\mathcal{I}} \, |\,  X_{(\mathcal{S})} \leq x_{(\mathcal{S})})   \notag \\
&=  P  (A_{\mathcal{I}}, X_{(\mathcal{I}^c)} \leq x_{(\mathcal{I}^c)}  \, |\, X_{(\mathcal{S})} \leq x_{(\mathcal{S})}  ) \notag \\
&=  P ( X_{(\mathcal{I})} < x_{(\mathcal{I})}, X_{(\mathcal{I}^c)} \leq x_{(\mathcal{I}^c)}  \, |\, X_{(\mathcal{S})} \leq x_{(\mathcal{S})} ),
\label{eq:IncExcl2}
\end{align}
where in the last equality we use the definition of $A$'s from~\eqref{eq:At}.  Now combining \eqref{eq:DeMorgan}-\eqref{eq:IncExcl2}, we have that
\begin{align}
\label{eq:JointProb}
&P (X_{(\mathcal{S})} = x_{(\mathcal{S})} \, | \,  X_{(\mathcal{S})} \leq x_{(\mathcal{S})}  ) =1 - \notag \\
& \sum_{t=1}^k  (-1)^{t-1}\!\! \sum_{\substack{\mathcal{I} \subseteq \mathcal{S} \\ |\mathcal{I}| =t} } \!\! P  ( X_{(\mathcal{I})} \!<\! x_{(\mathcal{I})}, X_{(\mathcal{I}^c)} \!\leq \! x_{(\mathcal{I}^c)} \, | \, X_{(\mathcal{S})} \! \leq \! x_{(\mathcal{S})}) .
\end{align}
We now note that the event in the conditioning in~\eqref{eq:JointProb}, namely, $X_{(\mathcal{S})} \leq x_{(\mathcal{S})}$, is a superset of the other event considered, $X_{(\mathcal{S})} = x_{(\mathcal{S})}$. It therefore follows that by multiplying both sides of~\eqref{eq:JointProb} by $P (X_{(\mathcal{S})} \leq x_{(\mathcal{S})}  )$, we obtain our probability of interest.  In other words,
\begin{align*}
&P(X_{(\mathcal{S})} \leq  x_{(\mathcal{S})}) P \left (X_{(\mathcal{S})} = x_{(\mathcal{S})} \, | \,  X_{(\mathcal{S})} \leq x_{(\mathcal{S})} \right ) \\
&= P \left (X_{(\mathcal{S})} = x_{(\mathcal{S})} \text{ and }  X_{(\mathcal{S})} \leq x_{(\mathcal{S})} \right ) \\
& = P(X_{(\mathcal{S})} =   x_{(\mathcal{S})}).
\end{align*}
Using the above in \eqref{eq:JointProb}, we find a representation for
 $P(X_{(\mathcal{S})} = x_{(\mathcal{S})})$ as:
\begin{align}
& P  (X_{(\mathcal{S})} = x_{(\mathcal{S})})   =P (X_{(\mathcal{S})} \leq x_{(\mathcal{S})}  ) \notag \\
&-\sum_{v=1}^k (-1)^{v-1} \sum_{\substack{\mathcal{I} \subseteq \mathcal{S} \\ |\mathcal{I}| =v} } P \left ( X_{(\mathcal{I})} < x_{(\mathcal{I})}, X_{(\mathcal{I}^c)} \leq x_{(\mathcal{I}^c)} \right ). \label{eq:JointProbNewProof}
\end{align}
We finally note that  the probability on the right side of \eqref{eq:JointProbNewProof} is equal to the result given in~\eqref{eq:JointProbNewa}, which can be seen by defining, $F_{(\mathcal{S})}(x_{(\mathcal{S})}) := P (X_{(\mathcal{S})} \leq x_{(\mathcal{S})} )$, and for all $\mathcal{I} \subseteq \mathcal{S}$,
\begin{align}
\label{eq:jointCDFGen}
F_{(\mathcal{I},\mathcal{I}^c)} \left(x^-_{(\mathcal{I})},x_{(\mathcal{I}^c)} \right) \! \!:= \!\!P\left(X_{(\mathcal{I})} \!\!<\!\! x_{(\mathcal{I})},X_{(\mathcal{I}^c)} \!\leq\! x_{(\mathcal{I}^c)} \right).
\end{align}
Now we discuss the results in \eqref{eq:JointProbNewb} and \eqref{eq:JointProbNewc}.
In words, the definition in \eqref{eq:jointCDFGen} implies that, for all $\mathcal{I} \subseteq \mathcal{S}$, the term $F_{(\mathcal{I}, \mathcal{I}^c)} ( x_{(\mathcal{I})}^-, x_{(\mathcal{I}^c)})$ is the probability that:
\begin{itemize}
\item For all $i_j \in \mathcal{I}$ with $j \in [k]$ there are at least $i_j$ observations less than $x_{(i_j)}$;
\item For all $i_t \in \mathcal{I}^c$ with $t \in [k]$ there are at least $i_t$ observations less than or equal to $x_{(i_t)}$.
\end{itemize}
Equivalently, we also note that $F_{(\mathcal{I}, \mathcal{I}^c)} ( x_{(\mathcal{I})}^-, x_{(\mathcal{I}^c)}  )$ can be computed as the probability that:
\begin{itemize}
\item For all $i_j \in \mathcal{I}$ with $j \in [k]$ there are at most $(n-i_j)$ observations  greater than or equal to $x_{(i_j)}$;
\item For all $i_t \in \mathcal{I}^c$ with $t \in [k]$ there are at most $(n-i_t)$ observations greater than $x_{(i_t)}$.
\end{itemize}
Thus, computing $P(X_{(\mathcal{S})} = x_{(\mathcal{S})})$ boils down to computing $F_{(\mathcal{I}, \mathcal{I}^c)} ( x_{(\mathcal{I})}^-, x_{(\mathcal{I}^c)} )$ for all subsets $\mathcal{I} \subseteq \mathcal{S}$.
Finally, simple counting techniques are used to show that $F_{(\mathcal{I}, \mathcal{I}^c)} ( x_{(\mathcal{I})}^-, x_{(\mathcal{I}^c)} )$ is equal to~\eqref{eq:JointProbNewb} with the function $g(\cdot,\cdot)$ is defined in~\eqref{eq:JointProbNewc}.
This concludes the proof of Lemma~\ref{lem:OS_dist}.
\end{proof}

\subsection{Continuous Random Variables}
\label{app:Sec:ProofOfJointPMFC}
We state a lemma from~\cite{arnold1992first} 
that computes the joint distribution of $k$ order statistics, and is the counterpart of Lemma~\ref{lem:OS_dist} for the case of continuous random variables.
\begin{lem} 
\label{lem:OS_cont}
Let $X_1, X_2, \ldots, X_n$ be i.i.d.\ r.v.\ from an absolutely continuous distribution with cumulative distribution function $F(x)$ and probability density function $f(x)$.
Let $ \mathcal{S} =\{ ( i_1,i_2,\ldots,  i_k ): 1 \le  i_1< i_2< \ldots <  i_k \le n  \}$ 
and 
\begin{align*}
f_{X_{(\mathcal{S})}}(x_{(\mathcal{S})}) = f_{X_{(i_1)},X_{(i_2)}, \ldots , X_{(i_k)}}(x_{(i_1)},x_{(i_2)}, \ldots , x_{(i_k)})
\end{align*}
be the joint probability density function of $X_{(\mathcal{S})}$, where $x_{(i)}$ denotes the observation associated to index $i$. Then, $f_{X_{(\mathcal{S})}}(x_{(\mathcal{S})})$ is non-zero only if
$
-\infty < x_{(i_1)} < x_{(i_2)} < \ldots < x_{(i_k)} < \infty,
$
and, when this is true, its expression is given by
\begin{align*}
&f_{X_{(\mathcal{S})}}(x_{(\mathcal{S})}) \\
&= g(n,i_{(\mathcal{S})}) \prod_{t=1}^k f(x_{(i_t)}) 
 \prod_{t=1}^{k+1} \left [F(x_{(i_t)}) - F(x_{i_{t-1}}) \right ]^{i_t-i_{t-1}-1},
\end{align*}
where $x_{(i_0)} = -\infty$, $x_{(i_{k+1})} = +\infty$, and, with $i_0 = 0$ and $i_{k+1} = n+1$,
\begin{align*}
g(n,i_{(\mathcal{S})}) = \frac{n!}{\prod_{t=1}^{k+1} (i_t - i_{t-1}-1)!}.
\end{align*}
\end{lem}

\section{Proof of Lemma~\ref{lemma:Binary}}
\label{app:BernEx}
First, for any $i \in [n]$, by Lemma~\ref{lem:OS_dist}, we have
\begin{align*}
P(X_{(i)}=0) &=  \sum_{k=i}^n  { n \choose k}  (1-p)^k p^{n-k}, \\
P(X_{(i)}=1) &=  1-P(X_{(i)}=0).
\end{align*} 
Thus,  $X_{(i)}$ is Bernoulli distributed with success probability $v(i)$, i.e., $X_{(i)} \sim \text{Ber}(v(i))$, where
\begin{align*}
&v(i) = 1-P(X_{(i)}=0) = 1 -  \sum_{k=i}^n  { n \choose k}  (1-p)^k p^{n-k} \notag \\
&= \sum_{k=0}^n  { n \choose k}  (1-p)^k p^{n-k} -  \sum_{k=i}^n  { n \choose k}  (1-p)^k p^{n-k}.
\end{align*}
Notice that $v(i) = P( B < i)$ where $B$ is a Binomial$(n, 1-p)$ random variable. 

We first consider the measure $\mathsf{r}_1(i,X^n)= H(X_{(i)})$ where the equality follows by Theorem~\ref{thm:RepresentationsCond}. Since $X_{(i)} \sim \text{Ber}(v(i))$, the entropy is given by 
\begin{align}
\mathsf{r}_1(i,X^n)=  H(X_{(i)})=h_b(v(i)) ,
\label{eq:EntropyOfOrderBernoulli}
\end{align}
where $h_b(t) := -t \log(t) -(1-t)\log(1-t)$ is defined to be the binary entropy function. 

Next, we focus on the metric $\mathsf{r}_3(i,X^n)$. By Theorem~\ref{thm:RepresentationsCond}, we have
$\mathsf{r}_3(i,X^n) =   \E[ (X_{(i)} -\E[X_{(i)} ] )^2]=  \mathsf{Var}(X_{(i)}).$
By the result just discussed, $X_{(i)} \sim \text{Ber}(v(i))$ and therefore
\begin{equation} 
\mathsf{Var}(X_{(i)}) = v(i)(1-v(i)) =  P( B < i)  P( B \ge i). \label{eq:VarianceInOrderBernoulli}
\end{equation}
Finally, we study the measure $\mathsf{r}_2(i,X^n)$. We have
\begin{align}
\mathsf{r}_2(i,X^n) &= \E\left[ \| \E[X^n]- \E[X^n| X_{(i)}]\|^2 \right] \notag  \\
&= \sum_{j=1}^n \E\left[ ( \E[X_j]- \E[X_j| X_{(i)}]  )^2 \right].
\label{eq:r_2_bernoulli}
\end{align}
Now consider just a single term inside the sum in \eqref{eq:r_2_bernoulli}:
\begin{align}
& \E\left[ ( \E[X_j]- \E[X_j| X_{(i)}] )^2 \right] = \E\left[( p- \E[X_j| X_{(i)}]  )^2 \right]  \notag \\
 &=  p^2 + \E\left[ (\E[X_j| X_{(i)}] )^2\right] - 2p \E\left[ \E[X_j| X_{(i)}] \right]  \notag \\
 &=  \E\left[ (\E[X_j| X_{(i)}] )^2\right] - p^2.
\label{eq:r_2_bernoulli_eq2}
\end{align}
Moreover, we notice that
\begin{align}
\E\left[ (\E[X_j| X_{(i)}] )^2 \right] &= P(X_{(i)} = 1)\left(\E[X_j| X_{(i)} = 1] \right)^2 \notag \\
& + P(X_{(i)} = 0) \left(\E[X_j| X_{(i)} = 0] \right)^2 .
\label{eq:r_2_bernoulli_eq3}
\end{align}
With the above in mind, we study the expectations $\E[X_j| X_{(i)} = 1] = P(X_j = 1| X_{(i)} = 1)$ and $\E[X_j| X_{(i)} = 0] = P(X_j = 1| X_{(i)} = 0)$. First, by Bayes rule,
\begin{align*}
P(X_j = 1| X_{(i)} = 1) &= \frac{P(X_{(i)} = 1 | X_j = 1) P(X_j = 1)}{P(X_{(i)} = 1)} \\
&= \frac{p \cdot P(X_{(i)} = 1 | X_j = 1) }{v(i)}.
\end{align*}
Now we study the probability $P(X_{(i)} = 1 | X_j = 1)$.  First notice that this equals the probability that there are at least $n-i+1$ total $1's$ in the sample $X^n$, given that $X_j = 1$, or in other words, this equals the probability that there are at least $n-i$ total $1's$ from the $n-1$ other sample values (excluding the $j^{th}$ one). Using this rationale,
\begin{align*}
&P(X_{(i)} = 1 | X_j =1) = \sum_{k=n-i}^{n-1} { n-1 \choose k}  (1-p)^{n-1-k} p^{k}\\
&=\sum_{k=0}^{i-1} { n-1 \choose k}  (1-p)^{k} p^{n-1-k} =P(B' < i),
\end{align*}
where $B' \sim \text{Binomial}(n-1, 1-p).$ Putting this all together, we have that
$\E[X_j| X_{(i)} = 1] =  \frac{p}{v(i)} P(B' < i).$
Similar reasoning, and the fact that $P(X_{(i)} = 0 | X_j = 1) = 1- P(X_{(i)} = 1 | X_j = 1)$, shows that
$\E[X_j| X_{(i)} = 1] =  \frac{p}{1-v(i)} P(B' \geq i ). $
Now, plugging the above results into the work in \eqref{eq:r_2_bernoulli}-\eqref{eq:r_2_bernoulli_eq3},
\begin{align*}
&\mathsf{r}_2(i,X^n) =  \frac{np^2}{v(i)}[ P(B' < i)]^2 +  \frac{np^2}{1-v(i)} [ P(B' \geq i)]^2 - np^2,
\end{align*}
where recall that $v(i) = P(B < i)$.

\section{Proof of Lemma~\ref{example:uniformContinious}}
\label{app:UnifExamp}
If $X_i's$ are i.i.d.\ $\sim \mathcal{U}(0,a)$, then $ \frac{1}{a}X_{(i)} \sim {\rm Beta}(i, n-i+1)$ with mean and variance  given by
\begin{align}
\E[X_{(i)}] &= \frac{a i}{n+1}, \ \text{and} \
\mathsf{Var}(X_{(i)})=  \frac{a^2 i (n+1-i)}{ (n+1)^2 (n+2)} .
\label{eq:beta_var}
\end{align} 
Thus, by Theorem~\ref{thm:RepresentationsCond}, we have
\begin{align*}
\mathsf{r}_3(i,X^n) &=   \E\left[ (X_{(i)} -\E[X_{(i)} ] )^2 \right] \\
&=  \mathsf{Var}(X_{(i)})
 = \frac{a^2 i (n+1-i)}{ (n+1)^2 (n+2)}.
\end{align*}
By taking the first derivative of $\mathsf{r}_3(i,X^n)$ above with respect to $i$ and equating it to zero, we obtain $i^\star_3(X^n)$ as in~\eqref{eq:i3starUn}.

We now compute  $ \mathsf{r}_2(i,X^n)$.  Using~\eqref{eq:m2_solvedcond}, we have
\begin{align}
\mathsf{r}_2(i,X^n) &= \E\left[ \| \E[X^n]- \E[X^n| X_{(i)}] \|^2 \right]  \notag \\
&= \sum_{j=1}^n \E\left[ ( \E[X_j]- \E[X_j| X_{(i)}] )^2\right].\label{eq:uniform_step1}
\end{align}
Now we look at computing the expectation $\E[X_j| X_{(i)} = x_{(i)} ]$.  By the law of total expectation,
\begin{align}
&\E[X_j | X_{(i)}= x_{(i)} ] \notag \\
&= \E \left[X_j | X_{(i)} = x_{(i)}, \{    X_j = X_{(i)} \} \right ] {P} (X_j  = X_{(i)} ) \notag\\
& \quad +   \E \left[X_j | X_{(i)}= x_{(i)}, \{   X_j <  X_{(i)}  \} \right ]{P} ( X_j <  X_{(i)}  ) \notag\\
&\quad +  \E \left[X_j | X_{(i)}= x_{(i)}, \{   X_j >  X_{(i)}  \} \right ] {P}( X_j >  X_{(i)} ).
\label{eq:lawoftotexp}
\end{align} 
Now we simplify the three terms of the above.  First notice that the probabilities can be computed using the fact that any $X_j$ is equally likely to produce the $i$-th order statistic, so
\begin{align*}
&{P} (X_j  = X_{(i)}) = \frac{1}{n}, \\
&{P} ( X_j <  X_{(i)}  ) =  \frac{i-1}{n},  \\
&{P} ( X_j >  X_{(i)} ) =  \frac{n-i}{n}.
\end{align*} 
Next we compute the expectations in~\eqref{eq:lawoftotexp}. Clearly, $ \E \left[X_j | X_{(i)} = x_{(i)}, \{    X_j =X_{(i)} \} \right ] = x_{(i)}$.  Moreover, we note that $X_j$ is independent of the event $\{X_{(i)}= x_{(i)}\}$ given $\{   X_j >  X_{(i)}  \}$ and hence
\[\E \left[X_j | X_{(i)}, \{   X_j >  X_{(i)}  \} \right ] = \E \left[X_j |    X_j >  x_{(i)} \right ] = \frac{a + x_{(i)}}{2}.\]
Similarly,
\[\E \left[X_j | X_{(i)}, \{   X_j <  X_{(i)}  \} \right ] = \E \left[X_j |  X_j <  x_{(i)} \right ]  = \frac{x_{(i)}}{2}.\]
Plugging these results into \eqref{eq:lawoftotexp}, we find
\begin{align}
2n \E[X_j | X_{(i)}= x_{(i)} ] &=   2 x_{(i)}   +    (i-1) x_{(i)} +  (n-i) (a+x_{(i)})  \notag\\
&=  (1 + n)x_{(i)}    + a (n - i) .
\label{eq:final_cond_exp}
\end{align} 
Now we use the result in~\eqref{eq:final_cond_exp} to simplify \eqref{eq:uniform_step1}.  
First,
\begin{align*}
&\mathsf{r}_2(i,X^n) = \sum_{j=1}^n \E\left[ ( \E[X_j]- \E[X_j| X_{(i)}] )^2 \right] 
\\&=   \sum_{j=1}^n   (\E[X_j])^2 - 2 \E[X_j]  \E\left[ \E[X_j| X_{(i)}] \right] +   \E\left[ (\E[X_j| X_{(i)}] )^2\right] \\
&=  -n  (\E[X_1])^2 +  \sum_{j=1}^n  \E\left[ (\E[X_j| X_{(i)}] )^2\right],
\end{align*}
where in the final equality we have used $ \E[X_j]  \E[ \E[X_j| X_{(i)}] ]  =  (\E[X_j])^2$ and that $ (\E[X_j])^2 =  (\E[X_1])^2$ for all $j \in [n]$.  Therefore, using that $n(\E[X_1])^2 = na^2/4$ and plugging the result in~\eqref{eq:final_cond_exp} into the above, we have
\begin{align*}
&\mathsf{r}_2(i,X^n) =  \frac{-na^2}{4} +  \frac{1}{4n}  \E \left[ \left(  (1 + n)X_{(i)}    + a (n - i)  \right)^2\right]  \\
&=   \frac{-na^2}{4} +  \frac{(n+1)^2}{4n}  \E  \left[  \Big(  X_{(i)}    + \frac{a (n - i)}{n+1}   \Big)^2  \right]  \\
&=  \frac{-na^2}{4} \\
&\quad +  \frac{(n+1)^2}{4n}  \left[\E  [  X^2_{(i)}]    + \Big(\frac{a (n - i)}{n+1}  \Big)^2 + 2 \E [ X_{(i)}] \frac{a (n - i)}{n+1} \right] .
\end{align*}
Next, note that by \eqref{eq:beta_var},
\begin{align*}
\E[X_{(i)}^2]  &= \mathsf{Var}[X_{(i)}] + (\E[X_{(i)}])^2  \\
&=  \frac{a^2 i(n+1-i)}{(n+1)^2(n+2)} + \frac{a^2 i^2}{(n+1)^2} =  \frac{a^2i(i+1)}{(n+1)(n+2)}.
\end{align*}
Therefore, using the above and $\E[X_{(i)}]=  a i/(n+1)$,
\begin{align*}
&\mathsf{r}_2(i,X^n) \\
&=  \frac{-na^2}{4} \\
&\quad +\!  \frac{(n+1)^2}{4n}  \left[ \frac{a^2i(i+1)}{(n+1)(n+2)}    \!+\! \frac{a^2 (n - i)^2}{(n+1)^2}  \!+\! \frac{2a^2 i (n - i)}{(n+1)^2} \right]  \\
&=    \frac{a^2  [ (n+1)i(i+1)   + (n+i)(n - i)(n+2)  - n^2(n+2)] }{4n(n+2)}\\
&=    \frac{a^2 i (n+1-i)}{ 4 n (n+2)} ,
\end{align*}
which has maximum value for $i^\star_2(X^n)$ as reported in~\eqref{eq:i3starUn}.

\bibliography{refs}
\bibliographystyle{IEEEtran}

\end{document}